%% file: main.tex
\documentclass[runningheads]{llncs}
\usepackage[T1]{fontenc}

\usepackage{commands-tam}
\usepackage{wrapfig}
\usepackage[utf8]{inputenc}
\usepackage{subcaption}
\usepackage{siunitx}
\usepackage{blkarray}
\usepackage{cite}
\usepackage[utf8]{inputenc}
\usepackage{hyperref}

\usepackage{mathtools}

\newcommand{\var}[1]{\texttt{#1}}

\usepackage{algorithm,algpseudocode}

\algnewcommand{\algorithmicforeach}{\textbf{for each}}
\algdef{SE}[FOR]{ForEach}{EndForEach}[1]
  {\algorithmicforeach\ #1\ \algorithmicdo}
  {\algorithmicend\ \algorithmicforeach}

\usepackage[textsize=scriptsize]{todonotes}

\newif\ifhighlight
\highlighttrue

\newif\ifshowtext
\showtextfalse

\begin{document}

\title{Self-Assembly of Patterns in the abstract Tile Assembly Model}

\author{Phillip Drake\inst{1} \and
Matthew J. Patitz\inst{1} \and
Scott M. Summers\inst{2} \and
Tyler Tracy\inst{1}}

\authorrunning{P. Drake, M.\,J. Patitz, S.\,M. Summers and T. Tracy}

\institute{University of Arkansas, Fayetteville, AR 72701, USA\\
\email{\{padrake,patitz,tgtracy\}@uark.edu} \and
University of Wisconsin-Oshkosh, Oshkosh, WI 54901, USA\\
\email{summerss@uwosh.edu}. This author was supported in part by University of Wisconsin Oshkosh Research Sabbatical (S581) Fall 2023.}

\maketitle

\begin{abstract}
In the abstract Tile Assembly Model, self-assembling systems consisting of tiles of different colors can form structures on which colored patterns are ``painted.'' We explore the complexity, in terms of the numbers of unique tile types required, of assembling various patterns.
We first demonstrate how to efficiently self-assemble a set of simple patterns, then show tight bounds on the tile type complexity of self-assembling 2-colored patterns on the surfaces of square assemblies. Finally, we demonstrate an exponential gap in tile type complexity of self-assembling an infinite series of patterns between systems restricted to one plane versus those allowed two planes.
\end{abstract}



\input{intro}

\input{prelims}

\input{simple-patterns}
\input{squares-tight-bounds}

\input{repeated-patterns}

\input{multilayer}

\bibliographystyle{splncs04}
\bibliography{tam,experimental_refs,slats}

\input{simple-patterns-appendix}
\input{squares-tight-bounds-appendix}
\input{repeated-patterns-appendix}
\input{multilayer-appendix}

\end{document}

%% file: intro.tex
\section{Introduction}\label{sec:intro}

During the process of self-assembly, a disorganized collection of components experiencing only random motion and local interactions combine to form structures. Examples of self-assembly abound in nature, from crystals to cellular components, and these systems have inspired researchers to study them to better understand the underlying principles governing them as well as to engineer systems that mimic them. Along the spectrum of both natural and artificial self-assembling systems are those which (1) use very small numbers of component types and form simple, unbounded, repeating patterns, (2) use a large number of component types, on the order of the entire sizes of the structures created, to form structures with highly-specified, asymmetric designs, and (3) those which use very small numbers of components to make arbitrarily large, bounded or unbounded, symmetric or asymmetric structures whose growth is directed algorithmically.

Systems in category (2), so-called \emph{fully-addressed} or \emph{hard-coded} have the benefit of being able to uniquely define each ``pixel'' of the structure and therefore to ``paint'' arbitrary pictures, which we'll refer to as \emph{patterns}, on their surfaces\cite{MonaLisa,Shih-OrigamiSlats,RothOrigami}.
Although it is easy to see that in models of tile-based self-assembly such as the abstract Tile Assembly Model (aTAM)\cite{Winf98} any finite structure or pattern can self-assemble from a hard-coded system, it has previously been shown that there exist infinite structures and patterns that cannot self-assemble in the aTAM \cite{jCCSA,jSSADST}.
The benefits of the \emph{algorithmic self-assembly}\cite{woods2019diverse,evans2014crystals,RothTriangles} of category (3) include precise formation of shapes using exponentially fewer types of components\cite{RotWin00,SolWin07}, thus reducing the cost, reducing the effort to fabricate and implement, and increasing the speed of growth of these systems\cite{crisscrossAccel}.
However, the drastic reduction in the number of component types means a corresponding increase in their reuse. This results in copies of the same component type appearing in many locations throughout the resulting target structure, thus removing the ability to uniquely address each pixel when forming patterns.
This generally results in a reduction in the number of patterns that are producible.

In this paper we study the trade-off between the numbers of unique components, or \emph{tile types}, needed to self-assemble designed patterns and the complexities of the patterns that can self-assemble. Past lines of work have dealt with the self-assembly of patterns in the aTAM from the perspective of the computational complexity of designing minimal tile sets to self-assemble given patterns (the so-called ``PATS'' problem)\cite{j2PATS,CzeizlerPopaPATS,MaLombardi08,LempCzeilOrp11} and others have shown the possibilities and impossibilities of assembling some classes of infinite patterns\cite{jCCSA,jSADS}. In contrast, our first results present constructions for making tile sets that self-assemble a series of relatively simple patterns to demonstrate how efficiently they can be built algorithmically. These consist of patterns of white and black pixels on the surfaces of squares, and include (1) a pattern with a single black pixel, (2) a pattern with some number $k$ of black pixels, and (3) grids of alternating black and white stripes. All of these are shown to have exponential reductions in tile type requirements, a.k.a. \emph{tile complexity}, over fully-addressed systems.

Our next pair of results combine to show a tight bound on the tile complexity of self-assembling arbitrary patterns of two colors on the surfaces of $n \times n$ squares for almost all patterns. Using an information theoretic argument we prove a lower bound, namely that for almost all patterns on $n \times n$ squares, such patterns have tile complexity $\Omega \left(\frac{n^2}{\log{n}}\right)$. We then provide a construction that, when given an arbitrary $n \times n$ pattern of black and white pixels, generates a tile set of $O\left(\frac{n^2}{\log{n}}\right)$ tile types that self-assemble an $n \times n$ square with black and white tiles that form that pattern. Although this is not a significant improvement over the $n^2$ tile types required to naively implement a fully-addressed set of tile types, the lower bound proves that for almost all patterns this is the best possible tile complexity.

Although the prior result showed that any $n \times n$ pattern can self-assemble using $O\left(\frac{n^2}{\log{n}}\right)$ tile types, our next result shows that, if given two planes in which tiles can self-assemble (one on top of the other) then it is possible for some patterns to self-assemble using exponentially fewer tile types than when systems are restricted to a single plane. (In fact, this result can be modified for arbitrary separation in tile complexity.) The proof of this result uses a novel application of diagonalization to tile-based self-assembly. Namely, one system simulates every system of a given tile complexity class for a bounded number of steps, sequentially and within a square of one plane, in order to algorithmically generate a pattern that is guaranteed not to be made by any of those systems, and then prints that generated pattern on the square of the second plane above the assembly that performed the simulations. We also show how to extend these square patterns infinitely to cover the plane, while maintaining the same tile complexity argument.

Overall, our results demonstrate boundaries on tile complexities of algorithmic self-assembling systems when forming patterns, and help to demonstrate their benefits over fully-addressed systems. To make some of our results easier to understand, we have created a set of programs and tile sets that can be used to view examples. These can be found online: \href{http://self-assembly.net/wiki/index.php/Pattern_Self-Assembly}{Pattern Self-Assembly Software} \cite{PatternAssemblySoftware}. Due to space limitations, proofs have been moved to a Technical Appendix in this version, and this full version can also be found online \cite{PatternsArxiv}.

%% file: prelims.tex
\section{Preliminary Definitions and Models}

In this section we define the terminology and model used throughout the paper.

\vspace{-15pt}
\subsection{The abstract Tile-Assembly Model}

We work within the abstract Tile-Assembly Model\cite{Winf98} in 2 and 3 dimensions. We will use the abbreviation \emph{aTAM} to refer to the 2D model, \emph{3DaTAM} for the 3D model, and \emph{barely-3DaTAM} to refer to the 3D model when restricted to the use of only 2 planes of the third dimension (a.k.a. the ``just barely 3D aTAM''), meaning that tiles can only be placed in locations with $z$ coordinates equal to $0$ or $1$ (use of the other two dimensions is unbounded). These definitions are borrowed from \cite{DDDIU} and we note that \cite{RotWin00} and \cite{jSSADST} are good introductions to the model for unfamiliar readers. 

Let $\mathbb{N}$ be the set of nonnegative integers, and for $n \in \mathbb{N}$, let $[n] = \{0, 1, ..., n-2, n-1\}$.
Fix $d\in\{2,3\}$ to be the number of dimensions and $\Sigma$ to be some alphabet with $\Sigma^*$ its finite strings. A \emph{glue} $g\in\Sigma^*\times\mathbb{N}$ consists of a finite string \emph{label} and non-negative integer \emph{strength}. There is a single glue of strength $0$, referred to as the \emph{null} glue. A \emph{tile type} is a tuple $t\in(\Sigma^*\times\mathbb{N})^{2d}$, thought of as a unit square or cube with a glue on each side. A \emph{tile set} is a finite set of tile types. We always assume a finite set of tile types, but allow an infinite number of copies of each tile type to occupy locations in the $\mathbb{Z}^d$ lattice, each called a \emph{tile}.

Given a tile set $T$, a \emph{configuration} is an arrangement (possibly empty) of tiles in the lattice $\mathbb{Z}^d$, i.e.\ a partial function $\alpha:\mathbb{Z}^d\dashrightarrow T$. Two adjacent tiles in a configuration \emph{interact}, or are \emph{bound} or \emph{attached}, if the glues on their abutting sides are equal (in both label and strength) and have positive strength. Each configuration $\alpha$ induces a \emph{binding graph} $B_\alpha$ whose vertices are those points occupied by tiles, with an edge of weight $s$ between two vertices if the corresponding tiles interact with strength $s$. An \emph{assembly} is a configuration whose domain (as a graph) is connected and non-empty. The \emph{shape} $S_\alpha \subseteq \mathbb{Z}^d$ of assembly $\alpha$ is the domain of $\alpha$. For some $\tau\in\mathbb{Z}^+$, an assembly $\alpha$ is \emph{$\tau$-stable} if every cut of $B_\alpha$ has weight at least $\tau$, i.e.\ a $\tau$-stable assembly cannot be split into two pieces without separating bound tiles whose shared glues have cumulative strength $\tau$. Given two assemblies $\alpha,\beta$, we say $\alpha$ is a \emph{subassembly} of $\beta$ (denoted $\alpha \sqsubseteq \beta$) if $S_\alpha \subseteq S_\beta$ and for all $p\in S_\alpha$, $\alpha(p)=\beta(p)$ (i.e., they have tiles of the same types in all locations of $\alpha$). 

A \emph{tile-assembly system} (TAS) is a triple $\calT=(T,\sigma,\tau)$, where $T$ is a tile set, $\sigma$ is a finite $\tau$-stable assembly called the \emph{seed assembly}, and $\tau\in\mathbb{Z}^+$ is called the \emph{binding threshold} (a.k.a. \emph{temperature}).
If the seed $\sigma$ consists of a single tile, i.e. $|\sigma| = 1$, we say $\calT$ is \emph{singly-seeded}.
Given a TAS $\calT=(T,\sigma,\tau)$ and two $\tau$-stable assemblies $\alpha$ and $\beta$, we say that $\alpha$ \emph{$\calT$-produces} $\beta$ \emph{in one step} (written $\alpha \to^{\calT}_1 \beta$) if $\alpha \sqsubseteq \beta$ and $|S_\beta \setminus S_\alpha| = 1$.
That is, $\alpha \to^{\calT}_1 \beta$ if $\beta$ differs from $\alpha$ by the addition of a single tile.
The \emph{$\calT$-frontier} is the set $\partial^{\calT}\alpha = \bigcup_{\alpha \to^{\calT}_1 \beta} S_\beta \setminus S_\alpha$ of locations in which a tile could $\tau$-stably attach to $\alpha$.
When $\calT$ is clear from context we simply refer to these as the \emph{frontier} locations.

We use $\mathcal{A}^T$ to denote the set of all assemblies of tiles in tile set $T$. Given a TAS $\calT=(T, \sigma, \tau)$, a sequence of $k\in\mathbb{Z}^+ \cup \{\infty\}$ assemblies $\alpha_0, \alpha_1, \ldots$ over $\mathcal{A}^T$ is called a \emph{$\calT$-assembly sequence} if, for all $1\le i < k$, $\alpha_{i-1} \to^{\calT}_1 \alpha_i$. The \emph{result} of an assembly sequence is the unique limiting assembly of the sequence. For finite assembly sequences, this is the final assembly; whereas for infinite assembly sequences, this is the assembly consisting of all tiles from any assembly in the sequence. We say that \emph{$\alpha$ $\calT$-produces $\beta$} (denoted $\alpha\to^{\calT} \beta$) if there is a $\calT$-assembly sequence starting with $\alpha$ whose result is $\beta$. We say $\alpha$ is \emph{$\calT$-producible} if $\sigma\to^{\calT}\alpha$ and write $\prodasm{\calT}$ to denote the set of $\calT$-producible assemblies. We say $\alpha$ is \emph{$\calT$-terminal} if $\alpha$ is $\tau$-stable and there exists no assembly that is $\calT$-producible from $\alpha$. We denote the set of $\calT$-producible and $\calT$-terminal assemblies by $\termasm{\calT}$. If $|\termasm{\calT}| = 1$, i.e., there is exactly one terminal assembly, we say that $\calT$ is \emph{directed}. When $\calT$ is clear from context, we may omit $\calT$ from notation.

\subsection{Patterns}

Let $C$ be a set of colors and let $P \subseteq (\mathbb{Z}^d \times C)$. We say that $P$ is a \emph{$d$-dimensional pattern}, i.e., a set of locations and corresponding colors. Let $\dom P$ be the set of locations that are assigned a color. A pattern is \emph{k-colored} when the number of unique colors used is $k$. Let $\var{Color}(P, \vec{l})$ be a function that takes a pattern and a location $\vec{l}$ and returns the color of the pattern at that location. (and is undefined if $\vec{l} \not \in \dom(P)$).

Given a TAS $\calT = (T,\sigma,\tau)$, we allow each tile type to be assigned exactly one \emph{color} from some set of colors $C$.
Let $C_P \subseteq C$ be a subset of those colors, and $T_{C_P} \subseteq T$ be the subset of tiles of $T$ whose colors are in $C_P$.
Given an assembly $\alpha \in \prodasm{\calT}$, we use $\dom(\alpha)$ to denote the set of all locations with tiles in $\alpha$ and $\dom_{C_p}(\alpha)$ to denote the set of all locations of tiles in $\alpha$ with colors in $C_P$.
Given a location $\vec{l} \in \mathbb{Z}^d$, let $\var{Color}(\alpha,\vec{l})$ define a function that takes as input an assembly and a location and returns the color of the tile at that location (and is undefined if $\vec{l} \not \in \dom(\alpha)$).
%
%
We say $\calT$ \emph{weakly self-assembles pattern $P$} iff for all $\alpha \in \mathcal{A}_{\Box}[\mathcal{T}]$, $\dom_{C_P}(\alpha) = P$ and $\forall (\vec{l},c) \in P, c = \var{Color}(\alpha,\vec{l})$.
We say $\calT$ \emph{strictly self-assembles pattern $P$} iff $T_{C_P} = T$, i.e. all tiles of $T$ are colored from $C_P$, and $\calT$ weakly self-assembles $P$ (i.e. all locations receiving tiles are within $P$).


Let the set of all patterns be $\mathbb{P}$, the set of all $c$-colored patterns be $\mathbb{P}_c$, and $\var{SQPATS}_{c,n} \subset \mathbb{P}_c$ be the set of all $c$-colored patterns that are on the surfaces of $n \times n$ squares. 
Let $\var{SQPATS}_{c} = \bigcup_{n \in \mathbb{Z}^+}{\var{SQPATS}_{c,n}}$, i.e., the set of all $c$-colored patterns that are on the surfaces of squares.
A \emph{pattern class} is an infinite set of patterns parameterized by some set of values $X$, and can be represented as a function $PC: X \rightarrow \mathbb{P}$ that maps parameters $X$ to some pattern $P \in \mathbb{P}$. Let $\mathbb{T}$ be the set of all aTAM systems. A \emph{construction for pattern class} $PC$ is a function $C_{PC}: \mathbb{P} \rightarrow \mathbb{T}$ that takes a pattern $P \in \mathbb{P}$ and outputs an aTAM system $\calT \in \mathbb{T}$ such that $\calT$ weakly self-assembles $P$.

%% file: simple-patterns.tex
\section{Simple Patterns}\label{sec:simple}

In this section, we define several relatively simple pattern classes and present constructions that can build them efficiently.

First we define a pattern class whose patterns each consist of a 2-colored $n \times n$ square that is completely white except for a single black pixel.
\begin{definition}[Single-Pixel Pattern Class]
    Given $n, i, j \in \mathbb{N}$, where $i,j < n$, define $\var{SinglePixel}(n, i, j) \rightarrow P$ such that (1) $P \in \var{SQPATS}_{2,n}$, (2) $\forall \vec{l} \in dom(P) - \{(i, j) \}$, $\var Color(P,\vec{l}) = \text{White}$, and (3) $\var Color(P, (i,j)) = \text{Black}$.
\end{definition}

\begin{theorem}\label{thm:single-pixel}
    For all $n, i, j \in \mathbb{N}$ such that $n \ge i,j$, there exists an aTAM system $\calT = (T,\sigma,2)$ such that $|\sigma|=1$, $|T| = O(\log(n))$, and $\calT$ weakly self-assembles $\var{SinglePixel}(n, i, j)$.
\end{theorem}

\begin{proof}
We present a construction that, given an $n,i,j \in \mathbb{N}$ such that $i,j < n$, creates a TAS $\calT$ with tile complexity of $O(\log n)$ that weakly self-assembles $P = \var{SinglePixel}(n,i,j)$. Figure \ref{fig:single-pixel} shows a high-level depiction of how we build the assembly. 

The assembly starts by growing a hard-coded rectangle with an empty interior called a \emph{counter box}. Each side of the box has glues for a counter to attach. The side lengths are $s = \lceil \log n \rceil$, as this is the maximum length needed to encode the bits required for a binary counter to count a full dimension of the entire square (which is the maximum that could be necessary). Thus, the counter box uses $O(\log n)$ tiles types. 

Each side of the counter box has a number encoded in the outward-facing glues. These numbers are pre-computed so that binary counter tiles (i.e. sets of tiles that operate as standard binary counters) grow outward from them to form a cross-like structure that extends to each boundary of the $n \times n$ square. A constant-sized set of tile types (independent of $i$,$j$, and $n$) is used for each of the four counters. To complete the $n \times n$ square, a constant-sized set of filler tiles fill in between the counters and inside the counter box.

The seed tile is a single black tile, from which the counter box grows, while all other tile types are white. If the black tile is within $\log n$ of the side of the square, then the counter box grows away from the edge, meaning the counter box always has room to grow inside of the $n \times n$ square.

We have shown that the tile complexity of this assembly is $O(\log n)$ and that it weakly self-assembles $P$. Thus, the Theorem \ref{thm:single-pixel} is proved.

\end{proof}

\begin{figure}
    \centering
    \begin{subfigure}{0.3\textwidth}
        \centering
        \includegraphics[width=1.0\textwidth]{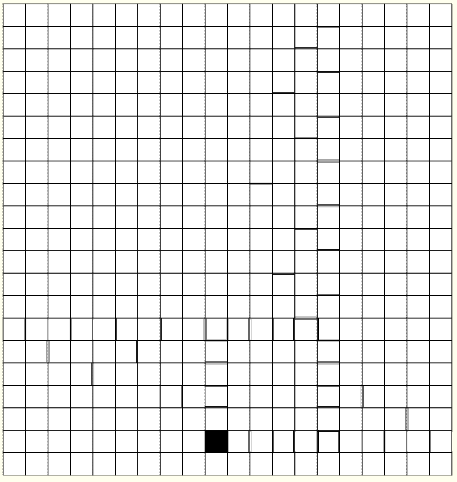}
        \caption{\label{fig:single-pixel-pattern}}
    \end{subfigure}
    \hspace{20pt}
    \begin{subfigure}{0.3\textwidth}
        \centering
        \includegraphics[width=1.0\textwidth]{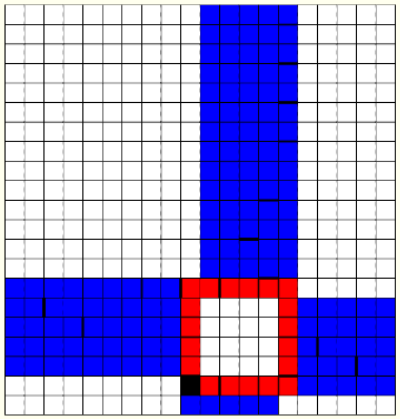}
        \caption{\label{fig:single-pixel-colored}}
    \end{subfigure}
    \caption{(a) An example of a single-pixel pattern. The black pixel is located at (10, 2). (b)  The same single-pixel pattern but with the counter box and counter tiles colored for demonstration. The counter box is colored red. The counters are colored blue. The white locations are filled by generic filler tiles.\label{fig:single-pixel}}
\end{figure}

Next, we define a pattern class whose patterns each consist of a 2-colored $n \times n$ square that is completely white except for a set of (separated and individual) black pixels. 

\begin{definition}[Multi-Pixel Pattern Class]
    Given $n \in \mathbb{N}$ and a set of locations $L \subseteq [0,n-1]^2$ such that $\forall (x, y), (x',y') \in L$, $|x - x'| \geq \lceil \log n \rceil \lor |y - y'| \geq \lceil \log n \rceil$, define $\var{MultiPixel}(n, L) \rightarrow P$ such that (1) $P \in \var{SQPATS}_{2,n}$, and (2) $\forall \vec{v} \in dom(P), \var Color(P,\vec{v}) =$ Black  if $\vec{v} \in L$ else White.
\end{definition}

\begin{figure}
    \centering
    \begin{subfigure}{0.4\textwidth}
        \centering
        \includegraphics[width=1.0\textwidth]{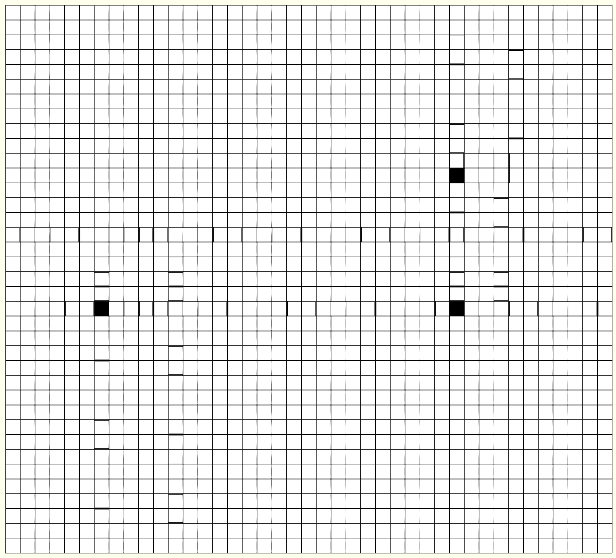}
        \caption{\label{fig:multi-pixel-pattern}}
    \end{subfigure}
    \hspace{20pt}
    \begin{subfigure}{0.4\textwidth}
        \centering
        \includegraphics[width=1.0\textwidth]{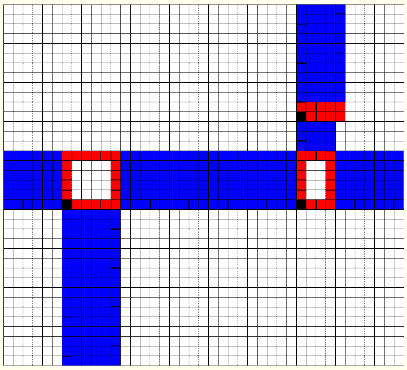}
        \caption{\label{fig:multi-pixel-colored}}
    \end{subfigure}
    \caption{(a) An example of a multi-pixel pattern with three black pixels. (b) A tree of counters is constructed to grow to each pixel and the edges of the square. The counter boxes are colored red. The counters are colored blue. The white locations are filled by generic filler tiles.\label{fig:multi-pixel}}
\end{figure}

\begin{theorem}\label{thm:multi-pixel}
    For all $n \in \mathbb{N}$ and sets of locations $L \subseteq [0,n-1]^2$ such that $\forall (x, y), (x',y') \in L$, $|x - x'| \geq \lceil \log n \rceil \lor |y - y'| \geq \lceil \log n \rceil$, there exists an aTAM system $\calT = (T,\sigma,2)$ such that $|\sigma|=1$, $|T| = O(|L|\log n)$, and $\calT$ weakly self-assembles $\var{MultiPixel}(n, L)$.
\end{theorem}

The proof of Theorem \ref{thm:multi-pixel} uses a construction that is an extension of that used in the proof of Theorem \ref{thm:single-pixel}, and essentially uses a series of counters to build a path connecting all pixels and filler tiles for the remaining portion of the square. A high-level depiction is shown in Figure \ref{fig:multi-pixel}, and full details can be found in Section \ref{sec:multi-pixel-appendix}.

The final (relatively) simple pattern class that we define contains patterns that each consist of a $2$-colored $n \times n$ square with a set of repeating black horizontal rows and a set of repeating black vertical columns.

\begin{definition}[Stripes Pattern Class]
    Given $n,i,j \in \mathbb{N}$, where $i,j < n$, define $\var{Stripes}(n,i,j) \rightarrow P$ such that (1) $P \in \var{SQPATS}_{2,n}$, and (2) $\forall x, y \in [0,n-1], \var Color(P, (x, y)) = \text{Black} \text{ if }x\mod i=0 \text{ or } y \mod j = 0, \text { else White}$.
\end{definition}

\begin{theorem}\label{thm:stripes}
    For all $n, i, j \in \mathbb{N}$, where $i,j < n$, there exists an aTAM system $\calT = (T,\sigma,2)$ such that $|\sigma|=1$, $|T| = O(\log(n))$ and $\calT$ weakly self-assembles $\var{Stripes}(n, i, j)$.
\end{theorem}

The proof of Theorem \ref{thm:stripes} can be found in Section \ref{sec:stripes-appendix}. It is done by construction, where the construction has counters that grow vertically to count to, and mark, the locations of horizontal strips, and counters that grow horizontally to count to, and mark, the locations of vertical stripes. Counters also keep track of the distance to the boundaries of the square to ensure growth stops at the correct locations. An overview can be seen in Figure \ref{fig:stripes}.
 
\begin{figure}
    \centering
    \begin{subfigure}{0.3\textwidth}
        \centering
        \includegraphics[width=1.0\textwidth]{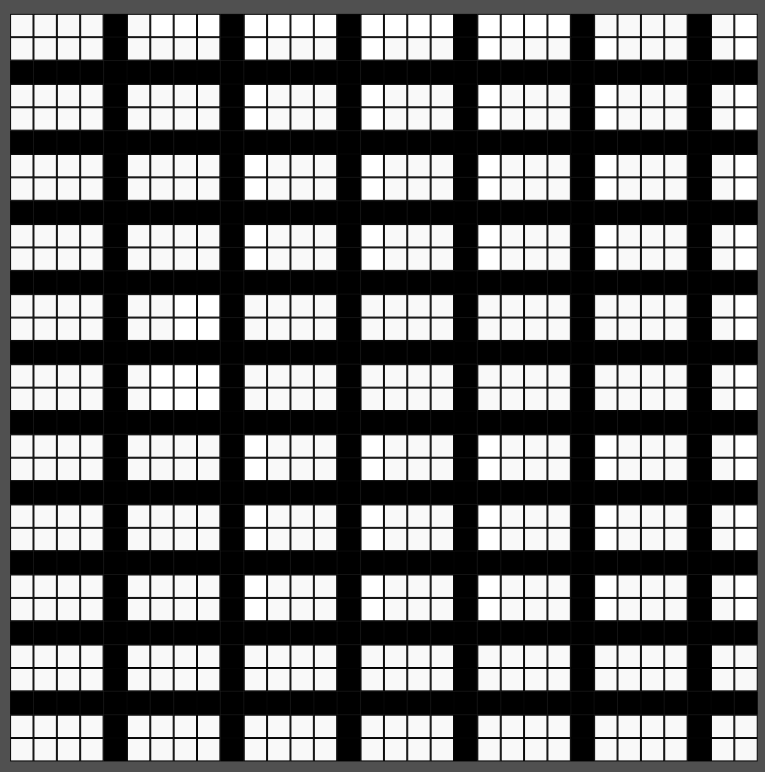}
        \caption{\label{fig:stripes-pattern}}
    \end{subfigure}
    \hspace{20pt}
    \begin{subfigure}{0.3\textwidth}
        \centering
        \includegraphics[width=1.0\textwidth]{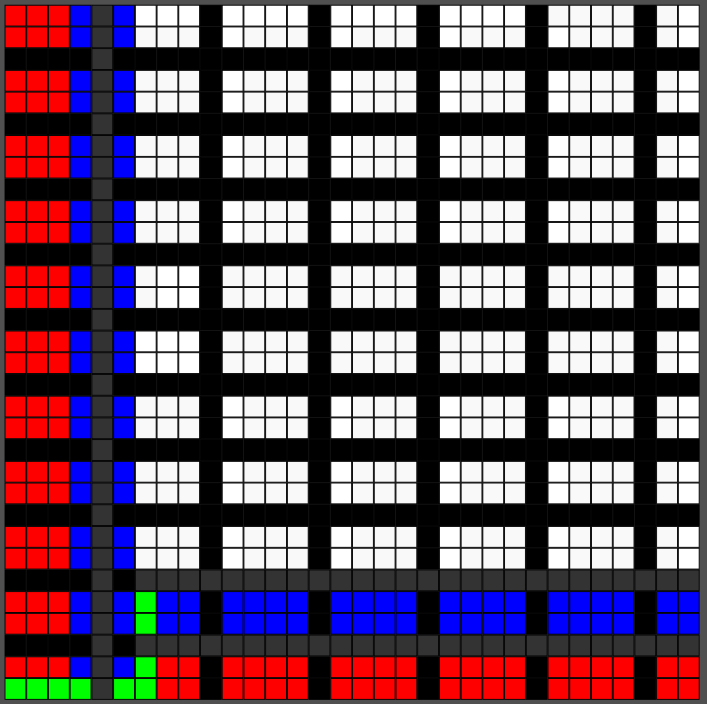}
        \caption{\label{fig:stripes-colored}}
    \end{subfigure}
    \caption{(a) An example of a stripes pattern. (b) The blue tiles count to the next stripe, while the red tiles count the number of stripes. Green tiles represent the starting rows for the counters (with the seed tile being at the corner where the green row and column intersect). Dark grey tiles represent counter tiles that are colored black\label{fig:stripes}}
\end{figure}

%% file: squares-tight-bounds.tex
\section{Tight Bounds for Patterns on $n \times n$ Squares}\label{sec:tight-squares}

In this section, we prove tight bounds on the tile complexity of self-assembling 2-colored patterns on the surfaces of $n \times n$ squares for almost all such patterns. 

\begin{theorem}\label{thm:tight-squares}
For almost all positive integers $n$ and  $P \in \var{SQPATS}_{2,n}$, the tile complexity of weakly self-assembling $P$ 
by a singly-seeded system in the aTAM is $\Theta\left(\frac{n^2}{\log n}\right)$.
\end{theorem}

We prove Theorem \ref{thm:tight-squares} by separately proving the lower and upper bounds, as Lemma \ref{lem:random-lower} and Lemma \ref{lem:random-upper}, respectively.

\begin{lemma}

For almost all patterns $P \in \var{SQPATS}_2$, the tile complexity of weakly self-assembling $P$ by a singly-seeded system in the aTAM is $\Omega\left(\frac{n^2}{\log n}\right)$.
\label{lem:random-lower}
\end{lemma}

The proof of Lemma~\ref{lem:random-lower} is a straight-forward information-theoretic argument and can be found in Section \ref{sec:random-lower-appendix}.

To prove the upper bound for Theorem~\ref{thm:tight-squares}, we prove the following, which is a stronger result that applies to all positive integers $n$.

\begin{figure}
    \centering
    \begin{subfigure}[T]{0.38\textwidth}
        \vspace{8pt}
        \includegraphics[width=\textwidth]{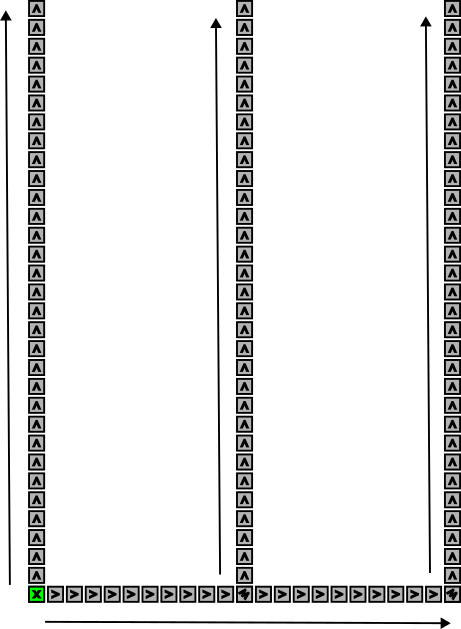}
        \caption{The skeleton. The seed is represented in green in the lower left and the arrows show the directions of growth.}
        \label{fig:skeleton}
    \end{subfigure}
    \hfill
    \begin{subfigure}[T]{0.50\textwidth}
        \includegraphics[width=\textwidth]{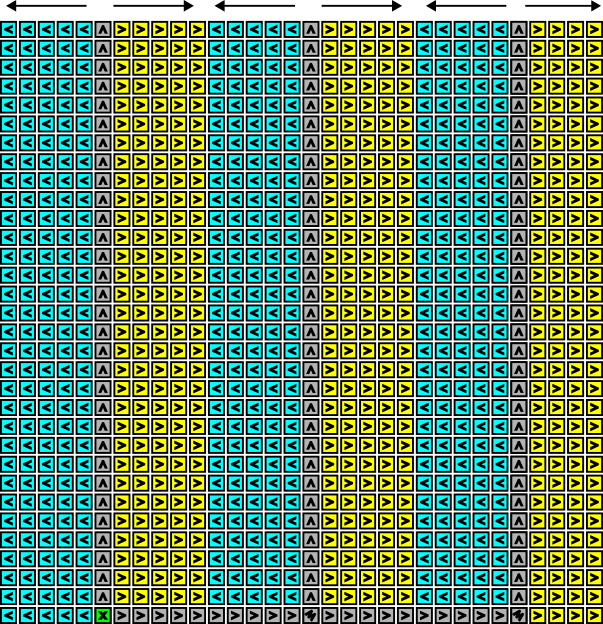}
        \caption{The square once the ribs of the skeleton have filled in (blue growing to the left, yellow growing to the right).}
        \label{fig:skeleton-filled}
    \end{subfigure}
    \caption{A schematic example of the construction of the proof of Lemma \ref{lem:random-upper}. Instead of showing the black and white colors corresponding to the pattern, we color the tiles to show the pieces of the construction to which they belong.}
    \label{fig:skeleton-construction}
\end{figure}

\begin{lemma}
For all positive integers $n$ and $P \in \var{SQPATS}_{2,n}$, there exists an aTAM system $\calT = (T,\sigma,1)$ such that $|\sigma|=1$, $|T| = O\left(\frac{n^2}{\log n}\right)$ and $\calT$ weakly self-assembles $P$.
\label{lem:random-upper}
\end{lemma}

\begin{proof}
We proceed by construction. Let $n \in \mathbb{Z}^+$ be the dimensions of the square and $P$ be the $n \times n$ pattern of black and white pixels to weakly self-assemble on the square.
Our construction will yield a system $\calT = (T,\sigma,1)$ that self-assembles an $n \times n$ square on which $P$ is formed by the black and white tiles of $T$. The tile set $T$ will be composed of two subsets, $T_s$ whose tiles form the \emph{skeleton}, and $T_r$ whose tiles form the \emph{ribs}.
We first explain the formation of the skeleton, then that of the ribs. Figure \ref{fig:skeleton-construction} shows a high-level depiction. 

\textbf{Skeleton} The seed is part of the skeleton and is placed at location $(\lfloor\log n\rfloor,0)$ and is given the color $\var{Color}(P, (\lfloor\log n\rfloor,0))$.
Since a vertical column of the skeleton has width one, and the ribs growing off of each side have length $\lfloor \log n \rfloor$, the width of a pair of ribs and its skeleton column (which we will call a \emph{rib-pair}) is $2\lfloor \log n \rfloor + 1$. Dividing the full width $n$ by the width of a rib-pair, and taking the floor, gives the number of full rib-pairs that will fit. Let $f = \lfloor\frac{n}{2\log n+1}\rfloor$ be this number.
Let $r = n \mod (2\lfloor \log n \rfloor + 1)$ be the remaining width after the last full rib-pair.
If $r < \lfloor \log n \rfloor + 1$, then a column of the skeleton grows up immediately to the right of the last full rib-pair, and its ribs are of length $r - 1$ and grow to the right.
If $r \ge \lfloor \log n \rfloor + 1$, then the last skeleton column grows upward $\lfloor \log n \rfloor$ positions to the right of the last full rib-pair and has full-length ribs (i.e., $\lfloor \log n \rfloor$) that grow to its left and ribs of length $r - (\lfloor \log n \rfloor + 1)$ grow to its right.
In the first case, the row of the skeleton that forms the bottom row of the square extends from the seed to $x$-coordinate $f(2\lfloor \log n \rfloor + 1) + 1$.
In the second case, that row extends from the seed to $x$-coordinate $f(2\lfloor \log n \rfloor + 1) + \lfloor \log n \rfloor + 1$.
The tiles of that row are hard-coded and there are $O(n)$ of them.
Starting with the seed and then occurring at every $2\lfloor\log n\rfloor + 1$ locations of the bottom row, a hard-coded set of tiles grows a column of height $n-1$. This row and set of columns are the full skeleton. The number of tile types is $O(n)$ for the row and $O(n)$ for each of the $O\left(\frac{n}{\log n}\right)$ columns, for a total of $O(n) + O\left(\frac{n^2}{\log n}\right) = O\left(\frac{n^2}{\log n}\right)$ tile types. Note that each skeleton tile type is given the color of the corresponding location in the pattern $P$.

\textbf{Ribs} From the east and west sides of each location on the columns of the skeleton, ribs grow. Each rib is composed of $\lfloor\log n\rfloor$ tiles (except the ribs growing from the easternmost column, which may be shorter). Since there are two possible colors for each of the $\lfloor \log n\rfloor$ locations of a rib, there are a maximum of $2^{\lfloor \log n \rfloor} \leq n$ possible color patterns for any rib to match the corresponding locations in $P$.
(Note that we will discuss the construction of the tiles for ribs that grow to the east, and for ribs that grow to the west the directions are simply reversed.)
For any given rib $r$, let the portion of $P$ corresponding to the locations of $r$ be represented by the binary string of length $\lfloor\log n\rfloor$ where each black location is represented by a $0$, and each white by a $1$.
For example, for a rib $r$ of length $5$ growing eastward from a column, if the corresponding locations of $P$ are ``black, black, white, black, white'', then the binary string will be ``$00101$''.
For each possible binary string $b$ of length $\lfloor\log n\rfloor$, i.e. $b \in \{0,1\}^{\lfloor\log n\rfloor}$, a unique tile type, $t_b$, is made.
with the glue $b$ on its west side and the glue $b[1:]$ (i.e. $b$ with its left bit truncated) on its east side.
This tile type is given the color corresponding to the first bit of $b$.
Additionally, for each skeleton column tile from which a rib should grow to the east with pattern $b$, the glue $b$ will be on its east side, allowing $t_b$ to attach.
This results in the creation of a maximum of $n$ unique tile types (and there will be another $n$ for the first tiles of each westward growing rib).
Recall that the tile types for the skeleton were already accounted for and each is hard-coded so that the placement of these glues does not require any new tile types for the skeleton.

Now, the process is repeated for each binary string from length $b-1$ to $1$, with the color of each tile being set to the value of the first remaining bit.
Each iteration requires half as many tile types to be created as the previous, i.e. $2^{\lfloor\log n\rfloor-1}$, then $2^{\lfloor\log n\rfloor-2}$, $\ldots$, $2$. Intuitively, each rib position has glues that encode their bit value in the pattern and the portion of the pattern that must be extended outward from them, away from the skeleton. Therefore, for the last position on the tip of each rib, there are exactly 2 choices, white or black, and so all ribs share from a set of two tile types made specially for the ends of ribs.
For the tile types of ribs that grow to the east, the total summation is $\Sigma^{\log n-1}_{x=0}{2^{\log n-x}} = 2n-2 = O(n)$. Accounting for the additional tile types needed for westward growing ribs, the full tile complexity of the ribs is $O(n)$.

Thus, the total tile complexity for the tile types of the skeleton plus those of the ribs is $O\left(\frac{n^2}{\log n}\right) + O(n) = O\left(\frac{n^2}{ \log n}\right)$.

\textbf{Correctness of construction} The system $\calT$ designed to weakly self-assemble $P$, as discussed, has a seed of a single tile, and since all tile attachments require forming a bond with a single neighbor, the temperature of the system can be $\tau=1$. Our prior analysis shows that the tile complexity is correct at $O\left(\frac{n^2}{\log n}\right)$, and showing that $\calT$ weakly self-assembles $P$ is trivial since (1) the tiles of the skeleton are specifically hard-coded to be colored for their corresponding locations in $P$, and (2) for each possible pattern corresponding to a rib there is a hard-coded set of rib tiles that match that pattern and grow from the skeleton into those locations. Thus, $P$ is formed and Lemma \ref{lem:random-upper} is proved, and with both Lemmas \ref{lem:random-lower} and \ref{lem:random-upper}, Theorem \ref{thm:tight-squares} is proved. (Example aTAM systems for this construction, as well as software capable of generating other systems for patterns derived from image files, and for simulating them, can be found online \cite{PatternAssemblySoftware}.)

\end{proof}

%% file: repeated-patterns.tex
\section{Repeated Patterns}\label{sec:repeated}

\begin{figure}
    \centering
    \includegraphics[width=2.3in]{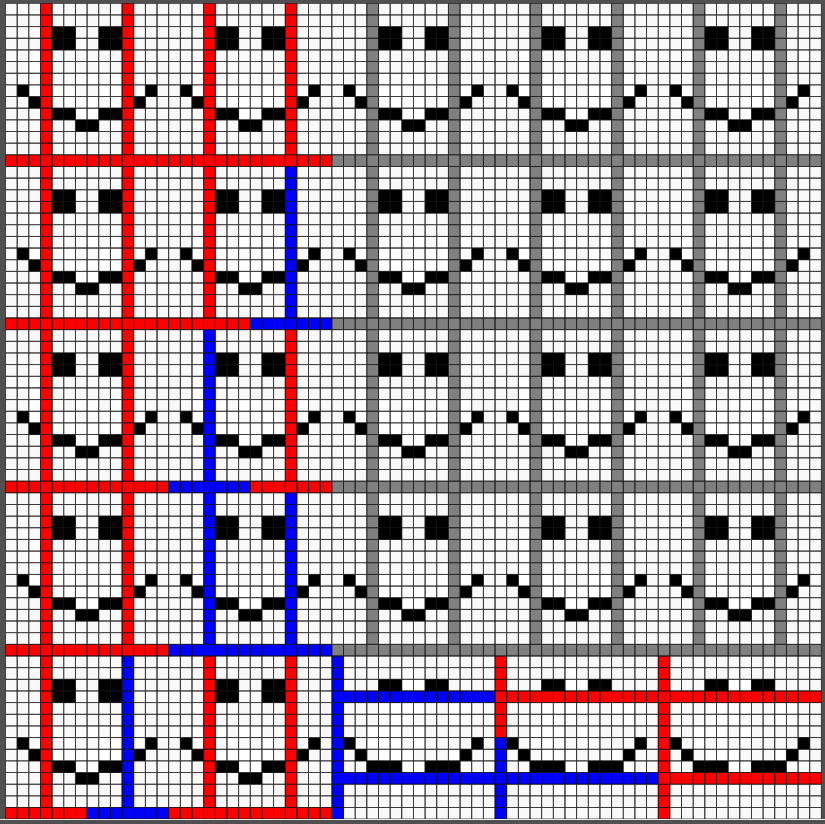}
    \caption{An example of an assembly that repeats a pattern $m = 5$ times horizontally and vertically. Each spine is colored solely for clarity of presentation, and in the actual construction, the colors of the tiles on the spines would match the pixels of the pattern. Red spines represent a 1, and blue spines represent a 0. The spines count upwards until the counter is finished}
    \label{fig:counter}
\vspace{-10pt}
\end{figure}

In this section, we discuss patterns consisting of repeated, arbitrary square sub-patterns, and that efficient systems exist that weakly self-assemble them.

\begin{definition}[Grid Repeat Pattern Class]
    Given $n,m \in \mathbb{Z}^+$ and $P \in \var{SQPATS}_{2,n}$, define $\var{GridRepeat}(P,m) \rightarrow P'$ such that $P' \in \var{SQPATS}_{2,nm}$ is an $nm \times nm$ square consisting of an $m \times m$ square composed of an $n \times n$ grid of copies of the pattern $P$.
\end{definition}

\begin{theorem}[Repeated Pattern Tile Complexity]
    For all $n, m \in \mathbb{N}$, and $P \in \var{SQPATS}_{2,n}$ there exists an aTAM system $\calT = (T,\sigma,2)$ such that $|\sigma|=1$, $|T| = O(\frac{n^2}{\log{n}} + \log{mn})$ and $\calT$ weakly self-assembles $\var{GridRepeat}(P, m)$.
    
    \label{thm:repeated-pattern}
\end{theorem}

The proof of Theorem \ref{thm:repeated-pattern} can be found in Section \ref{sec:repeated-appendix}. It makes use of an extension of the construction for the proof of Lemma \ref{lem:random-upper} and embeds a counter into the skeleton and ribs so that the copies of the sub-pattern are correctly counted.

%% file: multilayer.tex
\section{Barely-3DaTAM patterns}\label{sec:multilayer}





In this section, we show that there exist patterns, both finite and infinite, that can be weakly self-assembled using exponentially fewer tile types by barely-3DaTAM systems than by any regular, 2D aTAM systems. (We also note that the exponential separation can be increased arbitrarily.)

\begin{theorem}\label{thm:2layers-square}
    For all $n \in \mathbb{Z}^+$, for some $m \in \mathbb{Z}^+$ there exists a 7-colored $m \times m$ pattern, $p_n$, such that (1) no aTAM system $\calT_{\le n} = (T,\sigma,\tau)$ exists where $|T| \le n$, $|\sigma| = 1$, and $\calT_{\le n}$ weakly self-assembles $p_n$, but (2) a barely-3DaTAM system $\calT_{p_n} = (T_{p_n},\sigma_{p_n},2)$ exists where $ |T_{p_n}| = O(\log{n}/\log{\log{n}})$, $|\sigma_{p_n}| = 1$, and $\calT_{p_n}$ weakly self-assembles $p_n$.
\end{theorem}

\begin{figure}
\vspace{-10pt}
    \centering
    \begin{subfigure}[T]{0.45\textwidth}
        \includegraphics[width=\textwidth]{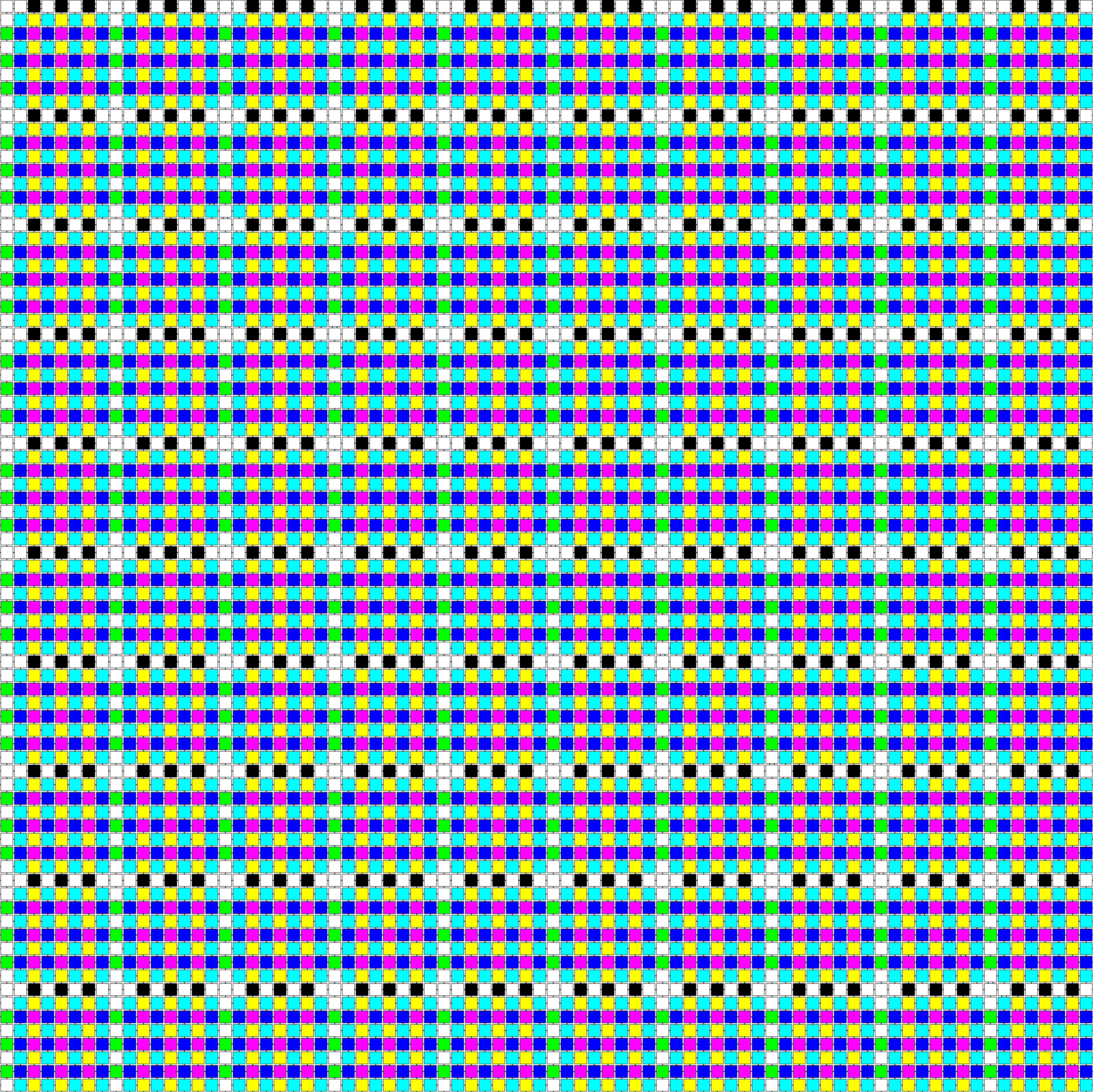}
        \caption{Example grid pattern for bit sequence $11010101$.}
        \label{fig:grid1}
    \end{subfigure}
    \hfill
    \begin{subfigure}[T]{0.45\textwidth}
        \includegraphics[width=\textwidth]{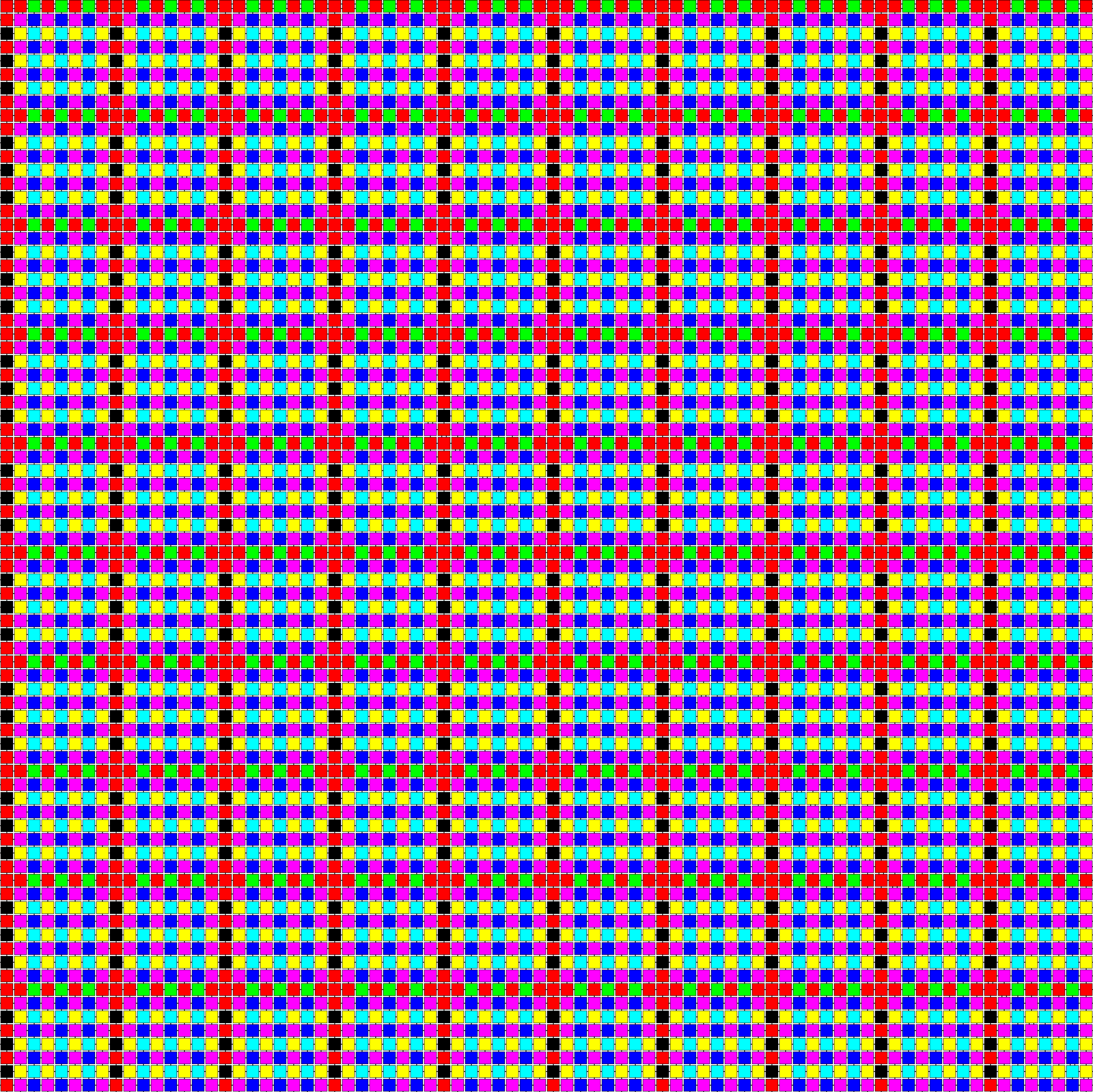}
        \caption{Example grid pattern for bit sequence $00101010$.}
        \label{fig:grid2}
    \end{subfigure}
    \caption{Example $p_n$ patterns created by the construction in the proof of Theorem \ref{thm:2layers-square}. The (repeatedly copied) binary sequence derived from the results of the simulations of aTAM systems starts at the top, with the two colors of that row, and all subsequent boundary rows, being determined by the first bit of that pattern. During the downward growth from that row, during which the full $m \times m$ square is formed, the repeating grid formed by the copies of that pattern is copied both downward and to both sides. The boundary columns also have two colors determined by the first bit of the sequence (one of them the same as in the boundary rows) for a total of 3 boundary colors. The interiors always use the same 4 colors.}
    \label{fig:grids}
\vspace{-25pt}
\end{figure}

\vspace{-10pt}
\paragraph{Proof sketch} (Here we give a sketch of the full proof of Theorem \ref{thm:2layers-square}. The full proof can be found in Section \ref{sec:multilayer-appendix}.) We prove Theorem \ref{thm:2layers-square} by giving the details of such a pattern $p_n$ that consists of a repeating ``grid'' of 7-colored lines on the surface of an $m \times m$ square, for $m \in \mathbb{Z}^+$ to be defined, and a barely 3D aTAM system $\mathcal{T}_{p_n} = (T_{p_n}, \sigma, 2)$ that weakly self-assembles $p_n$, with the tiles in $z=1$ colored in the pattern of $p_n$, and $|T_{p_n}| = O(\log{n}/\log{\log{n}})$ tile types.
We show that every 2D aTAM system $\mathcal{T}_{\le n}$ with $\le n$ tile types fails to weakly self-assemble $p_n$ by constructing $p_n$ so that it differs, in at least one location in each ``cell'' of a repeating grid of cells, from an assembly producible in each $\calT_{\le n}$. Two different examples of such patterns can be seen in Figure \ref{fig:grids}.
Each pattern $p_n$ consists of an $m \times m$ square that is covered in a repeating grid of square ``cells.'' Each cell is a $c \times c$ square (for $c \in \mathbb{Z}^+$, to be defined) where the north row and west column of each is considered ``boundary,'' and the rest of each cell is considered ``interior.''
The easternmost column and the southernmost row of cells may consist of truncated cells depending on the values of $m$ and $c$ (i.e., if $m \mod c \ne 0$).
Since each cell contributes a north and west boundary, each cell interior is completely surrounded by boundaries (except, perhaps, the easternmost column and southernmost row).
Depending on a bit sequence specific to each $p_n$ (to be discussed), the set of colors of the boundaries will be either $\{\var{White},\var{Green},\var{Black}\}$ or $\{\var{Red},\var{Green},\var{Black}\}$. The set of colors of the interiors will be $\{\var{Aqua},\var{Blue},\var{Yellow} ,\var{Fuchsia}\}$. Thus, each pattern $p_n$ will be composed of 7 colors.

The bit sequence that determines the colors used by the boundaries, and the ordering of the colors on the boundaries and in the interiors, is determined via simulations of a series of aTAM systems.
Intuitively, our proof utilizes a construction that performs a diagonalization against all possible aTAM systems with $\le n$ tile types by simulating each for a bounded number of steps, and for each keeping track of the color of tile it places in a location specific to the index of that system so that it can ultimately generate the colored pattern $p_n$ that differs in at least one location from every simulated system.

\begin{wrapfigure}{r}{0.5\textwidth}
\vspace{-15pt}
    \centering
    \includegraphics[width=2.3in]{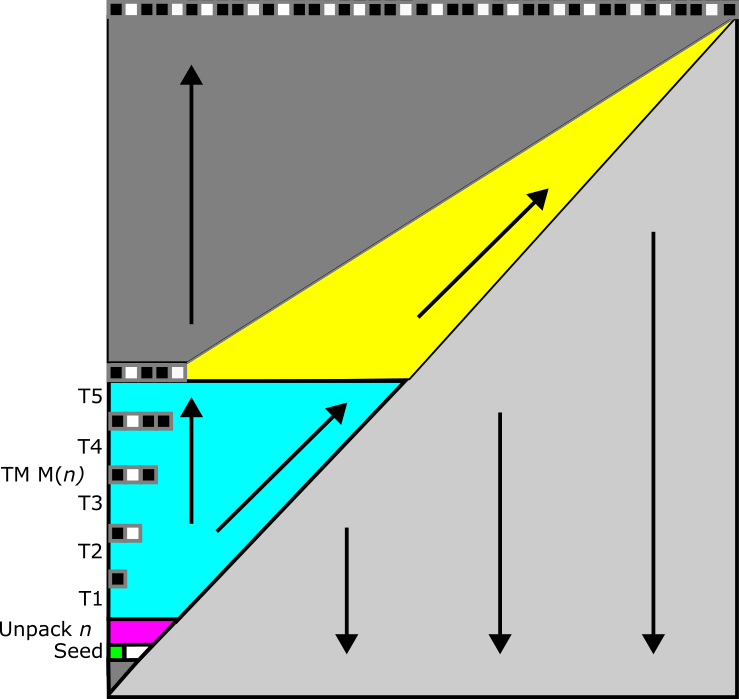}
    \caption{Schematic overview of the portion of the construction for the proof of Theorem \ref{thm:2layers-square} that grows in plane $z = 0$. Modules are not shown to scale. (Green) The seed tile, (Fuchsia) the base conversion module that unpacks the binary representation of $n$, (Aqua) the module that simulates the Turing machine $M$ on input $n$, which itself simulates each aTAM system with $\le n$ tile types in sequence and saves a result bit (\var{Black} or \var{White}) for each, (Yellow and Grey) the pattern of result bits is copied repetitively to the right until it covers the entire top row. (Light Grey) A ``filler'' tile causes the assembly to form a complete square.}
    \label{fig:wedge}
\vspace{-15pt}
\end{wrapfigure}

The dimensions of each $c \times c$ cell are $c = \var{SF}(n)$, where $\var{SF}(n)$ is a function that takes a number of tile types and returns an upper bound on the number of all possible singly-seeded aTAM systems with $\le n$ tile types and $\le 8$ colors. (Note that $\var{SF}(n)$ is actually greater than the number of such systems, and details of $\var{SF}(n)$ can be found in Section \ref{sec:multilayer-appendix}.) The colors of the rows and columns encode the bit sequence generated by the simulations, with the same bit sequence encoded in both the rows and the columns via an assignment of colors. There is a unique color assigned for each intersection of two bits (i.e. 00, 01, 10, and 11), with 4 colors reserved for boundaries of grid cells and 4 separate colors reserved for the interior locations of the grid cells. Therefore, the colored pattern of every cell differs from the assemblies produced by all aTAM systems with $\le n$ tile types. It forms on the top layer of a two-layered $m \times m$ square where $m =  O(n^{21n})$, and the barely-3DaTAM system that forms it uses $O(\log{n}/\log{\log{n}})$ tiles since the tiles types for all modules are constant except those that encode $n$ using optimal encoding \cite{AdChGoHu01}. (Note that the value of $m$ could be smaller, $O(n^{4n}n^8)$, if it wasn't desired that both planes be the same size. Additionally, by having $M$ simulate systems with larger tile sets, the value of $m$ would increase but the difference in tile complexity between the barely-3DaTAM system and the systems incapable of making its patterns could be increased beyond the current exponential bound.) Additional technical details, including pseudocode for the algorithms of the Turing machine $M$ and its simulations of all systems with $\le n$ tile types, the layout of data structures used during the simulation of a system, and time complexity analysis, can be found in Section \ref{sec:multilayer-appendix} of the technical appendix.

\vspace{-10pt}
\subsection{Extending a pattern $p_n$ to infinitely cover $\mathbb{Z}^2$}

Although it is already known that there are infinite patterns that can't weakly self-assemble from any finite-sized tile set \cite{jCCSA}, the following corollary simply shows how the previously defined patterns can be extended to infinitely cover the plane, while keeping an arbitrary spread in the tile complexity required by aTAM and barely-3DaTAM systems.

\begin{corollary}\label{cor:2layers-infinite}
    For all $n \in \mathbb{Z}^+$, there exists a 7-colored pattern, $p_{n_\infty}$, that infinitely covers the plane $\mathbb{Z}^2$ such that (1) no aTAM system $\calT_{\le n} = (T,\sigma,\tau)$ exists where $|T| \le n$, $|\sigma| = 1$, and $\calT_{\le n}$ weakly self-assembles $p_{n_\infty}$, but (2) a barely-3DaTAM system $\calT_{p_n} = (T_{p_n},\sigma_{p_n},2)$ exists where $ |T_{p_n}| = O(\log{n}/\log{\log{n}})$, $|\sigma_{p_n}| = 1$, and $\calT_{p_n}$ weakly self-assembles $p_{n_\infty}$.
\end{corollary}

To prove Corollary \ref{cor:2layers-infinite}, we extend the construction from the proof of Theorem \ref{thm:2layers-square} so that every pattern $p_n$ from the proof of Theorem \ref{thm:2layers-square} is extended to infinitely cover the $\mathbb{Z}^2$ plane, becoming $p_{n_\infty}$, by usage of ``grid-reconstruction,'' i.e., a method of copying the square grid infinitely to each side. This requires $O(1)$ unique tile types in addition to those used in the previous construction. Due to the symmetry exhibited by all $m \times m$ squares of $p_n$ patterns along their northeast $\rightarrow$ southwest diagonal, copies of the same pattern may be copied along these diagonals infinitely. Details of the construction can be found in Section \ref{sec:multilayer-infinite-appendix}.
\vspace{-20pt}

%% file: simple-patterns-appendix.tex
\newpage

\section{Technical Details of the Proofs of Section \ref{sec:simple}}\label{sec:simple-appendix}

In this section we include technical details of the constructions used in the proofs of Section \ref{sec:simple}. 

\subsection{Proof of Theorem \ref{thm:multi-pixel}}\label{sec:multi-pixel-appendix}

\begin{proof}
    We present a construction that, given a square dimension $n$ and set of valid locations $L$, provides a TAS $\calT$ with a tile complexity of $O(|L|\log n)$ that weakly self-assembles $\var{MultiPixel}(n, L)$. Figure \ref{fig:multi-pixel} shows a high-level depiction of how such an assembly is built. 

    We first inspect $L$ and determine a path of horizontal and vertical segments such that the path visits each pixel in $L$. The path can not intersect itself, but it can branch. This can be done for any set of points contained in the $n \times n$ square as long as they meet the criteria of all being separated by a distance of at least $\lceil \log n \rceil$. The path must also touch each side of the square. Every point where the path changes direction, or there is a black pixel, is called a \emph{node}. For each node, a counter box (like that in the construction for Theorem \ref{thm:single-pixel}) grows.

    The counter box for the node at the beginning of the path contains the seed tile (i.e., the seed tile is one of the tiles of that counter box, and intiaties its growth). Each segment of the path has a unique set of tiles that build a binary counter that builds a rectangle the length of that segment. For every node on the path, a unique counter box is built so that its sides encode the counts to the next nodes in each direction. The counter boxes of the nodes containing black pixels will each have a black tile for the corresponding location. The counter boxes will never overlap since the points are all $\ge \lceil \log n \rceil$ distance apart. 

    White filler tiles fill the insides of the counter boxes and in the locations bounded on two (or more) sides by the binary counters forming the segments of the path. Since the path touches all boundaries of the $n \times n$ square, the entire square is filled. Since each location in $L$ gets a black tile in the correct location of the corresponding counter box, and the other tiles are white, the assembly forms the entire pattern.

    There are $|L|$ counter boxes, each requiring $O(\log n)$ tiles types for the $4$ sides of the box (each using $\lceil \log n \rceil$ hard-coded tiles) and $O(\log n)$ tile types for the $4$ counters that grow from those sides. There is a constant number of filler tiles. Thus the tile complexity of $\calT$ is $O(|L|\log n)$ and, since $\calT$ weakly self-assembles the pattern $\var{MultiPixel}(n,L)$, Theorem \ref{thm:multi-pixel} is proved.     
\end{proof}

\subsection{Proof of Theorem \ref{thm:stripes}}\label{sec:stripes-appendix}

\begin{proof}
    We present a construction that given $n, i, j \in \mathbb{N}$, where $i,j < n$, outputs an aTAM system $\calT$ that weakly self-assembles $\var{Stripes}(n,i,j)$. Figure \ref{fig:stripes} shows a high-level depiction of the construction. 

    We start by creating four binary counters: two that count vertically upward, $v_1$ and $v_2$, and two that count horizontally to the right, $h_1$ and $h_2$. The counter $v_1$ counts upwards from $0$ to $i$ with the tiles being white except for those of the $i$th row, which are black. The counter $v_2$ grows to the left of $v_1$ and only increments after each $i$th row. At every $i$th row, $v_1$ causes $v_2$ to increment and, if $v_2$ hasn't reached its maximum, $v_2$ causes $v_1$ to reset so it can again count upward to $i$. The counter $v_2$ counts the number of stripes in the assembly, which equals $\lfloor \frac{n}{i} \rfloor$.  The other two counters, $h_1$ and $h_2$ perform the same actions but horizontally and make every $j$th column black, while the rest are white. The two pairs of columns grow two rectangles that form an 'L'. On the right side of the rightmost tiles of $v_1$ and the top side of the topmost tiles of $h_1$, the glues expose the black and white patterns determined by those counters. A constant-sized set of tiles cooperatively grow in the corner they form, extending that black and white pattern to fill out the rest of the square with that pattern.

    If $n \mod i \ne 0$, then once $v_2$ reaches its maximum value it initiates the growth of a final tile set that is a binary counter counting the remaining distance. This is analogously done by $h_2$. These extra counters use $O(\log n)$ tile types that are all white.

    If $\log j > i$ or $\log i > j$, then the counter would be longer than the first cell and flow into the first stripe. In this case, we add extra counter tiles that are painted black. We hard-code the starting row of the counters to have a black tile at the position of the stripe. This tile will have a different glue such that only the black version of the counter tiles can attach to that position of the counter row. This allows the counter to operate still while building the stripe. 

    The seed for this construction is a single tile that the starting row for the $v_1$ and $v_2$ counters, and the starting column for the $h_1$ and $h_2$ counters, attach to.

    This construction needs 4 counters with hard-coded initial rows and a constant number of tile types to perform the counting, for a total of $O(\log n)$ tile types. There is a constant number of filler and stripe tiles. The seed tile is a single tile that initiates growth of the starting values for all the counters. Thus, Theorem \ref{thm:stripes} is proved. 
\end{proof}

%% file: squares-tight-bounds-appendix.tex
\section{Technical Details of the Proofs of Lemma \ref{lem:random-lower}}\label{sec:random-lower-appendix}

In this section we include the proof of Lemma \ref{lem:random-lower}.

\begin{proof}
%

%
%
Let $\mathcal{T} = (T,\sigma,\tau)$ be a TAS such that $|\sigma|=1$, $\tau \in \mathbb{Z}^+$ is a fixed constant, and $B\subseteq T$ is the subset of tile types that are black (while the others are white), $n \in \mathbb{Z}^+$, and assume $\mathcal{T}$ weakly self-assembles an arbitrary $P \in \var{SQPATS}_{2,n}$. Without loss of generality, we can assume that the strength of every glue of $T$ is bounded by $\tau$ (since any glue with strength $\ge \tau$ can be replaced by one with strength $\tau$ without changing the behavior of the system).
%
%
%
%

%
Given that $P \in \var{SQPATS}_{2,n}$, let $w = w_{n^2-1}\cdots w_0 \in \{0,1\}^{n^2}$ be a bit sequence of length $n^2$ corresponding to the white and black pixels of $P$.
%
%
%
%
It is easy to see that for every $n \in \mathbb{Z}^+$ and $w \in \{0,1\}^{n^2}$, there exists a TAS $\mathcal{T}_w = \left(T, \sigma, \tau\right)$ such that $\left| T\right| = O\left(n^2\right)$.
Going forward, let $n \in \mathbb{Z}^+$ and $w \in \{0,1\}^{n^2}$ be arbitrary and suppose $\mathcal{T}_w$ is the corresponding TAS that weakly self-assembles the pattern corresponding to $w$.

Note that $\mathcal{T}_w$ has $4\left| T_w \right|$ glues, each strength is bounded by $\tau$, which is a fixed constant, and every tile is either in $B$ or not. 
This means $\mathcal{T}$ can be represented using $O\left( \left|T_w\right| \log \left|T_w\right|\right)$ total bits.
Let $\left \langle \mathcal{T}_w \right \rangle$ be such a representation of $\mathcal{T}_w$.
Let $w \in \{0,1\}^*$, and $U$ be a fixed universal Turing machine.  
The Kolomogorov complexity of $w$ is: $K_U(w) = \min\left\{ \left| \pi \right| \mid U(\pi) = w  \right\}$.
In other words, $K_U(w)$ is the size of the smallest program that when simulated on $U$ outputs $w$.
Let $m$ be a non-negative integer and $\varepsilon > 0$ be a fixed real constant.
The number of binary strings of length less than $m - \varepsilon m$ is at most $1 + 2 + 4 + \cdots + 2^{\lfloor m- \varepsilon m \rfloor - 1} = 2^{\lfloor m -  \varepsilon m \rfloor}-1 < 2^{\lfloor m-\varepsilon n\rfloor} < 2^{m - \varepsilon m + 1}$.
Define $A_{m,\varepsilon} = \left\{ w \in \{0,1\}^m \mid K_U(w) \geq (1-\varepsilon)m \right\}$.
Note that 
$$
\frac{\left| A_{m,\varepsilon} \right|}{2^m} \geq \frac{2^m - 2^{m - \varepsilon m + 1}}{2^m} = 1 - \frac{2^{m - \varepsilon m + 1}}{2^m}= 1 -  \frac{1}{2^{\varepsilon m - 1}}.
$$
Thus, we have
$$
\lim_{m \rightarrow \infty}{\frac{\left| A_{m,\varepsilon} \right|}{2^m}} = 1,
$$
which means that if $\varepsilon > 0$ is fixed, then for almost all strings, $w \in \{0,1\}^m$, $(1-\varepsilon)n < K_U(w)$.
There exists a fixed program $\pi_{SA}$ that takes as input $\langle \mathcal{T}_w \rangle$, simulates it, and outputs the string $w$ that corresponds to the pattern $P_w$ that $\mathcal{T}_w$ self-assembles.
Then, for almost all strings $w \in \{0,1\}^{n^2}$, $$
(1-\varepsilon)|w| < K_U(w) < C_1\left( \left| \pi_{SA} \right| + \left|T_w\right| \log \left|T_w\right|\right) < C_2 \left| T_w \right| \log |w|.
$$
It follows that $\left| T_w \right| = \Omega\left( \frac{|w|}{\log |w|}\right) = \Omega\left( \frac{n^2}{\log n^2} \right) = \Omega\left( \frac{n^2}{\log n}\right)$.

Since $n$ and $w$ were arbitrary, Lemma \ref{lem:random-lower} follows.

\end{proof}

%% file: repeated-patterns-appendix.tex
\section{Technical Details of the Construction for Theorem \ref{thm:repeated-pattern}}\label{sec:repeated-appendix}

In this section we provide the details of the proof of Theorem \ref{thm:repeated-pattern}.

\begin{proof}
    We will prove Theorem \ref{thm:repeated-pattern} by presenting a construction for $\var{GridRepeat}$ and showing that the produced systems have $O(\frac{n^2}{\log{n}} + \log{nm})$ tiles. 

    Let $\calT_{P}$ be a TAS that weakly self-assembles $P$ using the skeleton construction from \ref{thm:tight-squares}. We will use this assembly as a sub-component of the larger assembly that builds $\var{GridRepeat}(P, m)$. We divide the skeleton into $O(\frac{n}{\log n})$ parts called spines. Each spine is divided into two parts: the shaft and the base. The shaft is the vertical bar that extends northwards through the assembly. The ribs attach to the shaft on the east and west sides. We call the length of the ribs that connect to the east and west $r_e$ and $r_w$, respectively. The base of the spine is a subsection of the south row of the skeleton extending as far as the ribs that connect to the spine. It is centered at the bottom of the shaft and has length $r_e + r_w + 1$. The skeleton construction will also be modified to use temp 2 and have the bottom row extend the length of the full square. This does not change the construction's tile complexity.  

    Now, we present the construction. Given a pattern $P$ and an integer $m$, a TAS $\calT_R$ is built.  Let $T_{COUNT}$ be a constant set of counter tile types. For each tile type in $T_{COUNT}$, we create a copy of each spine from $\calT_P$. The new spine will have the glues from the counter tile appended to the north of the shaft and the south, east, and west ends of the base. A copy of the rib tiles will also be made with the counter tile's north glue appended to their glue. These extra rib tiles are necessary to propagate the top glue of the spine to the left and right of the top of the shaft (see Figure \ref{fig:counter-binding} to see how new spines attach).
    This allows for the entire spine to be constructed and other spines to grow off of it according to the growth of the counter tiles. A new spine might start growing from the north of the previous shaft, or from a the east or west of spines that have already attached. We repeat this process again to create a version of $\calT_P$ where the structure is rotated 90 degrees so it grows to the east, but the pattern is still facing upwards. This will grow to the east end of the square. 

    Then, another copy of the tile types from $\calT_P$ is added to fill the rest of the square with the pattern. At this point, there are a constant number of modified copies of the tile types of $\calT_P$ in the new tile set. Since Lemma \ref{lem:random-upper} showed the tile complexity of creating an arbitrary square pattern such as $\calT_P$ to be $O(\frac{n^2}{\log n})$, the overall tile complexity at this point therefore remains $O(\frac{n^2}{\log n})$.

    The seed of $\calT_R$  is a single tile at the southwest corner of the assembly. From this tile, a hardcoded set of tiles called $t_s$ will grow.  $t_s$ will consist of the bases of $O(\log m)$ spines all connected horizontally to each other. $t_s$ is longer than a single instance of $P$ when $\log m > \frac{n}{\log n}$. The north side of the tiles in $t_s$ are encoded with glues that bind to the base of the shaft of the spines. The glue corresponds to the north glues of the tiles in the base row of a counter assembly. The counter will be set to count to $m$. $t_s$ also consists of tiles that grow off the end of the horizontal row to grow a vertical column of tiles that encode a counter in a similar way. $t_s$ has $O(\log nm)$ tile types since it consists of $O(\log m)$ spine base rows and each of those are $O(\log n)$ tiles. 

    From $t_s$, instances of $P$ will grow northward and eastward forming an L shape that is $m$ by $m$. The fill copy of $\calT_P$ fills in the spaces between the L, forming the entire pattern. 
    
    \textbf{Correctness of construction}\\
    The construction outputs a TAS $\calT_R$ that is designed to weakly self-assemble $\var{GridRepeat}(P, m)$ for all patterns and values of $m$. It has a seed of a single tile and grows the entire skeleton from it. Our prior analysis shows that the tile complexity is correct at $O(\frac{n^2}{\log n} + \log{nm})$. Thus Theorem~\ref{thm:repeated-pattern} is proved.
\end{proof}

\begin{figure}
    \centering
    \includegraphics[width=3.0in]{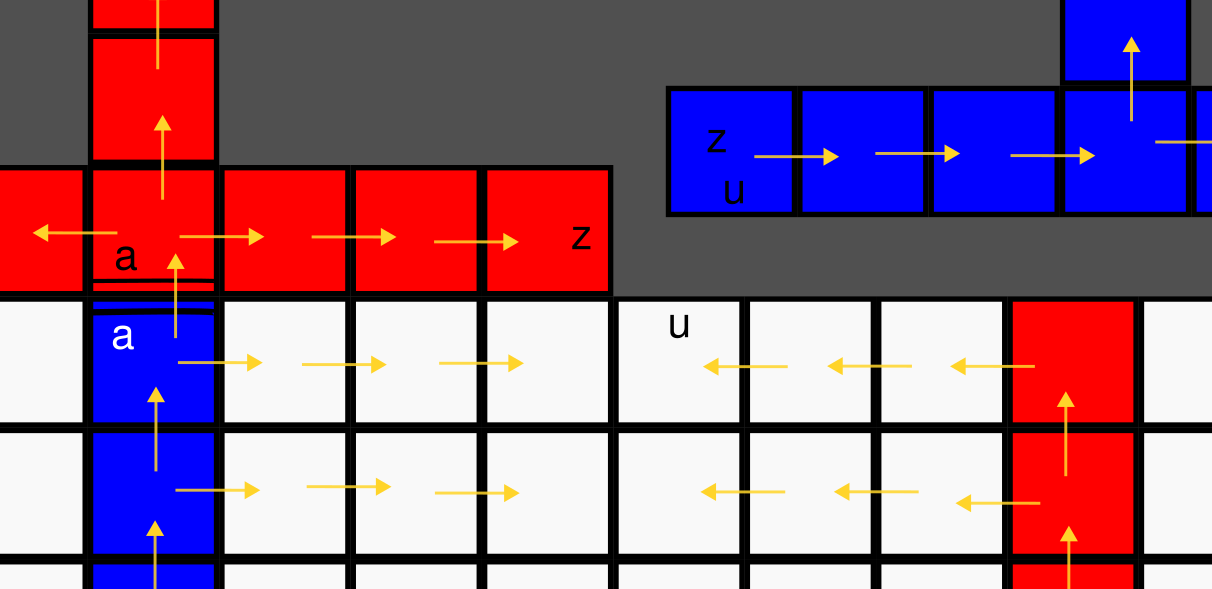}
    \caption{An example of how the binding of a spine works. The arrows indicate the growth of the tiles. The letter indicates glues. In this example, after the blue spine at the bottom finishes growing, the next red spine can grow.  The incoming blue spine tiles cooperatively bind with this red spine and the ribs from the spine below. This tile will grow the rest of the spine, and the pattern will continue. }
    \label{fig:counter-binding}
\end{figure} 

%% file: multilayer-appendix.tex
\section{Technical Details of the Construction for Theorem \ref{thm:2layers-square}}\label{sec:multilayer-appendix}

In this section we include technical details of the construction used to prove Theorem \ref{thm:2layers-square}.

At a high-level, the construction consists of a handful of components (of varying complexity) that can be seen schematically depicted in Figure \ref{fig:wedge}. The number $n$ is encoded in $O(\log{n}/\log{\log{n}})$ tile types following the technique of \cite{AdChGoHu01}. From the seed tile of $\calT_{p_n}$, the $O(\log{n}/\log{\log{n}})$ tiles representing $n$ in an optimally compressed base grow to the right. Then, rows grow upward and to the right to do a base conversion in which the bits of $n$ are ``unpacked'' so that the northern glues of the tiles of top row of that triangle represent $n$ as $\log{n}$ bits (shown in fuchsia in Figure \ref{fig:wedge}). (Figure \ref{fig:bit-unpacking} depicts a slightly more detailed example of the bit unpacking.) A simple set of filler tiles grow to the south of the seed's row to form the bottom of the triangle. The tile complexity of this stage of the construction is  $O(\log{n}/\log{\log{n}})$.
The tile complexity of the remaining portions of the construction is $O(1)$, as they use a constant number of tile types independent of $n$.

\begin{figure}
    \centering
    \includegraphics[width=3.5in]{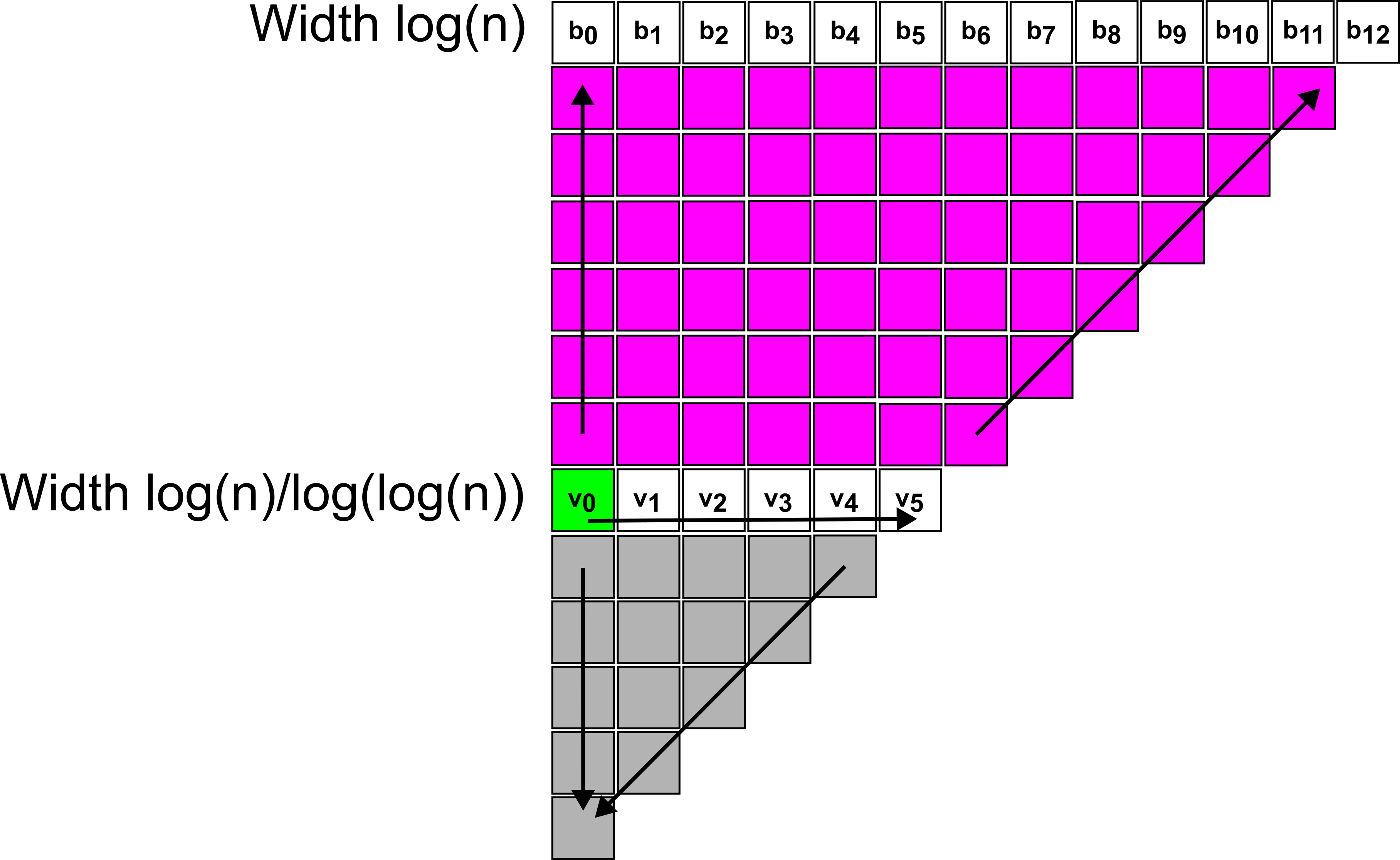}
    \caption{Schematic example of the growth of the base conversion module's growth from Figure \ref{fig:wedge}. The seed tile is shown in green. Using the technique of \cite{AdChGoHu01}, the number $n$ is encoded using $O(\log{n}/\log{\log{n}})$ tile types. These tile types form the (white) row that grows to the east of the seed. A base conversion then occurs via rows that grow to the north (fuchsia) to convert $n$ to binary, so that the northern row of this module (white) consists of tiles that encode $n$ in binary.}
    \label{fig:bit-unpacking}
\end{figure}

\subsection{Simulation of all aTAM systems with $\le n$ tile types}\label{sec:simulation}

A zig-zag Turing machine module (i.e., a standard aTAM construction in which a growing assembly simulates a Turing machine while rows grow in alternating, zig-zag, directions and each row increases in length by one tile - see \cite{jSADS,DirectedNotIU,Versus} for some examples) uses $n$ as input. The Turing machine $M$ simulates all singly-seeded aTAM systems containing $\le n$ tile types, $\le 8$ colors, and single-tile seeds, each for a bounded amount of time to be discussed. (The portion of the assembly that simulates $M$ is depicted in aqua in Figure \ref{fig:wedge}).

We now note that two aTAM systems $\calT_1 = (T_1, \sigma_1, \tau_1)$ and $\calT_2 = (T_2, \sigma_2, \tau_2)$ could be identical except that the strength of every glue in $\calT_1$ is doubled in $\calT_2$ and $2\tau_1 = \tau_2$. These are technically different aTAM systems, and an infinite set of such systems could be made by an infinite number of such strength and temperature doublings. If they all have $n$ tile types, it would then seem impossible to simulate, in finite time, the infinite set of systems with $\le n$ tile types. However, we note that there exists a color-preserving bijection (i.e., one that only maps tile types of the same color to each other) $f: T_1 \rightarrow T_2$ such that 
$\mathcal{A}[f(\calT_1)] = \mathcal{A}[\calT_2]$ and $\mathcal{A}_{\Box}[f(\calT_1)] = \mathcal{A}_{\Box}[\calT_2]$ (i.e. under the mapping of $f$ they have the exact same producible and terminal assemblies). Thus, these systems are \emph{equivalent} (in terms of assemblies produced), and simulating only one of any set of equivalent systems is necessary.
This is because our pattern $p_n$, by differing from the assembly produced by one such system will differ from the patterns produced all equivalent systems.
Furthermore, in \cite{chen2015program}, they developed the notion of \emph{strength-free} aTAM systems and used them to show that there are a bounded number of equivalence classes of aTAM systems (i.e. sets of equivalent systems) given a fixed number of tile types. A strength-free system, rather than assigning strength values to glues and a temperature value to the system, instead assigns \emph{cooperation sets} to tile types to define which subsets of their sides are sufficient for tile attachment. In this way, they abstract away the notions of glue strength and temperature and capture the behaviors of tiles, and they show how to both enumerate all possible strength-free systems with $n$ tile types and how to convert each into a standard aTAM system if possible. (For some strength-free systems, there is no valid corresponding aTAM system, and their algorithm can accurately report when that is the case.)

The Turing machine $M$ of our construction makes use of the tools of \cite{chen2015program} to first compute the number of singly-seeded strength-free systems with $\le n$ tile types and $\le 8$ colors as follows:
\begin{enumerate}
    \item A tile type has at most $168$ different possible cooperation sets (see \cite{chen2015program}).
    \item Each tile type can be one of 8 colors.
    \item For each tile side there are at most $n$ glue labels, or the \emph{null} glue to choose from. (If $m > n$ unique glue labels appear on the same side of the tile types in a set of at most $n$ tile types, then at least $m - n$ of them must not match glues on the opposite side of any tile type and therefore can be replaced by the \emph{null} glue without changing behavior.) 
    \item There are at most $(n+1)^4 = 4n^4$ ways to assign glue labels to the $4$ sides.
    \item Encoding each tile as a list of $4$ glue labels, a color, and a cooperation set yields $4n^4*8*168= 5376n^4$ tile types.
    \item The number of tile sets with $n$ tile types taken from the full set of $5376n^4$ is therefore $(5376n^4)^n$.
    \item Since each tile in a tile set could be the seed of a unique system, there are $(5376n^4)^{n+1}$ different strength-free systems with $n$ tile types.
    \item Thus, there are at most $\Sigma_{i=1}^n (5376i^4)^{i+1} \le (5376n^4)^{n+2}$ strength-free systems with at most $n$ tile types. We will refer to this number as $\var{SF}(n)$. (Note that in our pseudocode implementation of the algorithm of $M$, the value of $\var{SF}(n)$ is slightly smaller as it counts a bit more efficiently than this crude approximation, but that does not change the correctness of the discussion nor the asymptotic bounds.)
\end{enumerate}

\begin{figure}
    \centering
    \includegraphics[width=2.8in]{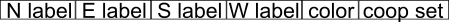}
    \caption{Layout of the binary representation of a strength-free tile type. Each of the ``N label,'' ``E label,'' ``S label,'' and ``W label'' fields are one of $n+1$ numbers from $0$ to $n$. The ``color'' field is one of $8$ numbers from $0$ to $7$, and ``coop set'' is one of $168$ values from $0$ to $167$.}
    \label{fig:strength-free-tile}
\end{figure}


To build pattern $p_n$, a bit sequence of length $\var{SF}(n)$, that we'll call $b_n$, will be generated by computing and saving a single bit for each of the $\var{SF}(n)$ strength-free systems (ultimately allowing the $p_n$ to differ from each of them). We will use the value of $\var{SF}(n)$ to determine the number of steps for which each system must be simulated, as discussed later.

By Theorem 3.1 of \cite{chen2015program} there is an algorithm that, given a strength-free aTAM system with $\le n$ tile types as input, returns an equivalent standard aTAM system if one exists, or \var{False} if not, in time $O(n^5)$. We'll call the function that implements this $\var{GetEquivalentATAMSystem}$. For each $0 \le i < \var{SF}(n)$, strength-free system $\mathcal{S}_i$ will be given to $\var{GetEquivalentaTAMSystem}$ which will either (1) convert it to a standard aTAM system $\calT_i$ that will next be simulated for a bounded time so that a bit value can be computed from it and saved as the $i$th bit of $b_n$, or (2) if $\mathcal{S}_i$ does not have an implementable aTAM system (and thus $\var{GetEquivalentaTAMSystem}$ returns \var{False}), the default bit value of $0$ will be saved as the $i$th bit of $b_n$.

Recall that $p_n$ consists of a grid of cells of size $c \times c$ repeated horizontally and vertically. For the value of $c$, we use $\var{SF}(n)$. As each system $\calT_i$, derived from $\mathcal{S}_i$, for $0 < i \le \var{SF}(n)$, is simulated, its index $i$ is noted so that the construction can guarantee that the $i$th row and $i$th column of each $\var{SF}(n) \times \var{SF}(n)$ cell will differ from a (potentially) corresponding location in the assembly produced by $\calT_i$ during its simulation. (Again noting that if there is no corresponding $\calT_i$ for some $\mathcal{S}_i$ we just save the bit $0$ since it's a ``don't care'' location.)
We want to ensure that for each $\calT_i$, if it happens to make a grid cell of $\var{SF}(n) \times \var{SF}(n)$ tiles composed of 7 colors (recall that although 8 colors are allowed in our construction, any given $p_n$ will only use 7 of them), $p_n$ differs in at least one location in each of its cells from at least one location of a cell produced by $\calT_i$. Any system $\calT_i$ that does not even produce a single valid cell of size $\var{SF}(n) \times \var{SF}(n)$ bounded by the boundary colors has no chance of generating $p_n$ so can be easily discounted and again a ``don't care'' bit of $0$ can be saved for its index in $b_n$. For all other $T_i$, a bit computed after running the simulation of $T_i$ is used to ensure $p_n$ differs.

\subsection{Making $p_n$ differ from each simulated system}\label{sec:differ}

For the simulation of each system, we do not impose a restriction upon the translation of the pattern that the system makes relative to our target pattern $p_n$ whose southwestern corner is at location $(0,0,0)$. Therefore, we cannot assume the relative position of the seed tile of $\calT_i$ with respect to any portion of $p_n$, and we simulate each $\calT_i$ until it grows an assembly that has at least one dimension (width or height) that spans the distance of a full cell with boundaries on both sides. To do so, we simulate each $\calT_i$ for $4\var{SF}(n)^2$ steps because this is the number of tiles contained within a $2 \times 2$ square of grid cells, ensuring that irrespective of the position of the seed tile with respect to the pattern formed, the full dimension of at least one grid cell and its boundaries in that dimension must be spanned. Figure \ref{fig:grid1-trimmed} shows an example of grid cells and a bounding box of that size, demonstrating why such bounds suffice. (Note that any system that becomes terminal before reaching such a size clearly cannot weakly self-assembly $p_n$, which is much larger, so we can record a ``don't care'' output bit of $0$ for that simulation.) Since we are only concerned with systems that can create grid-like patterns with the same boundary and interior colors used by $p_n$ (as differing patterns created by other systems will immediately disagree with $p_n$), we can inspect the assembly $\alpha_i$ produced by the simulation of $\calT_i$ as follows.

\begin{figure}
    \centering
    \includegraphics[width=2.5in]{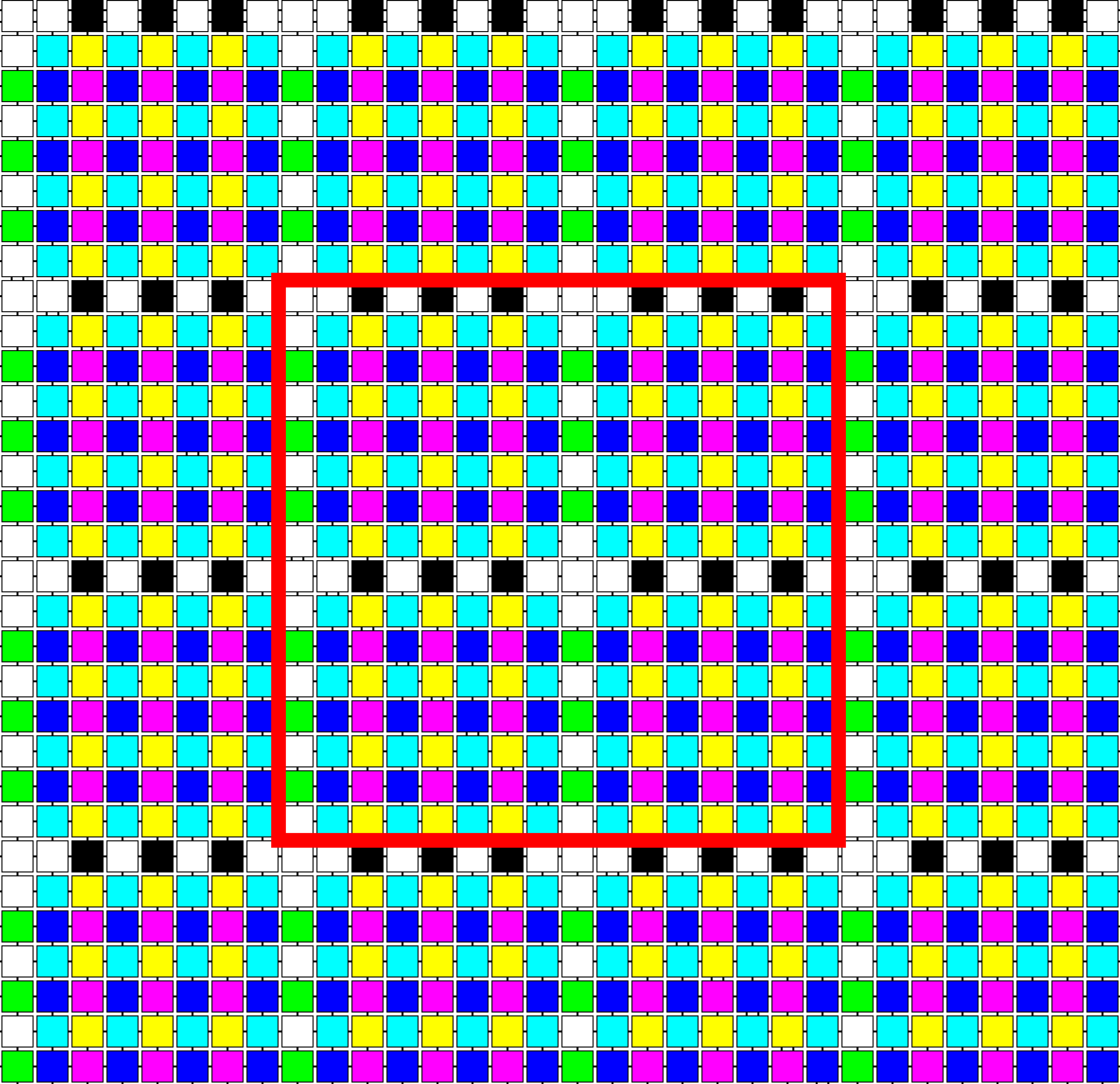}
    \caption{A portion of a pattern $p_n$ for the bit sequence $11010101$, showing a $4 \times 4$  grid of cells and a bounding box (Red) enclosing a $2 \times 2$ portion of the grid, containing $4\var{SF}(n)^2$ tiles, where $\var{SF}(n)$ is the width and height of a cell. Any assembly containing $4\var{SF}(n)^2$ tiles must contain a connected component of at least width $\ge \var{SF}(n)+1$ and/or height $\ge \var{SF}(n)+1$, ensuring two boundary locations on the sides of an assembly spanning a cell.}
    \label{fig:grid1-trimmed}
\end{figure}


Without loss of generality, assume that $\alpha_i$ has width $\ge 2\var{SF}(n)$. If not, it must have height $\ge 2\var{SF}(n)$ and the algorithm searches from north to south instead of west to east as described below.
Starting from any leftmost tile of $\alpha_i$, we inspect its color and continue inspecting the colors of tiles one position to the east of the previous, looping until we encounter one whose color is a boundary color.
(A simple example can be seen in Figure \ref{fig:assembly-indexed}.)
At that point, we skip an additional $i$ tiles to the east (noting that the $y$-coordinates don't matter, only the $x$-coordinates). Once a tile is found at that $x$-coordinate, which is guaranteed by the number of tiles in $\alpha_i$ and the assumption that its width (rather than height) is $\ge 2\var{SF}(n)$, its color is noted.
The pattern $p_n$ will be produced so that the colors of each column represent a bit, 0 or 1, and the colors of each row represent a bit, 0 or 1.
For the locations of a boundary row or column, there are 4 possible colors ($\{\var{White},\var{Green},\var{Black},\var{Red}\}$) used to represent the intersection of each location's row and column value, i.e., 00, 01, 10, and 11.
For the other (a.k.a. interior) locations, a set of 4 different colors is used ($\{\var{Aqua},\var{Blue},\var{Yellow}\,\var{Fuchsia}\}$). Indexing each row and column by $0 < i < \var{SF}(n)$ allows the $i$th row and $i$th column of each grid cell to be associated with a bit value. 
The bit value chosen to be saved is determined by the analysis of the color of the tile at index $i$ in $\alpha_i$ (i.e., the tile at an $x$-coordinate that is $i$ greater than that of a tile with a boundary color). If the color of the tile there is one of the two colors associated with a boundary row representing a 0, or one of the two colors associated with an interior row representing a 0, the bit value 1 is saved as the $i$th bit of $b_n$. Otherwise, the bit value 0 is saved. When the pattern $p_n$ is later produced, it will use colors associated with this ``flipped'' bit for all rows and columns at index $i$ of all cells of the grid. Therefore, every tile of $p_n$ that is $i$ locations to the east of any tile with a boundary color will have a different color than the tile of $\alpha_i$ that is $i$ positions to the east of a tile with a boundary color. In this way, the pattern produced by $\calT_i$ cannot be $p_n$.

\begin{figure}
    \centering
    \includegraphics[width=1.8in]{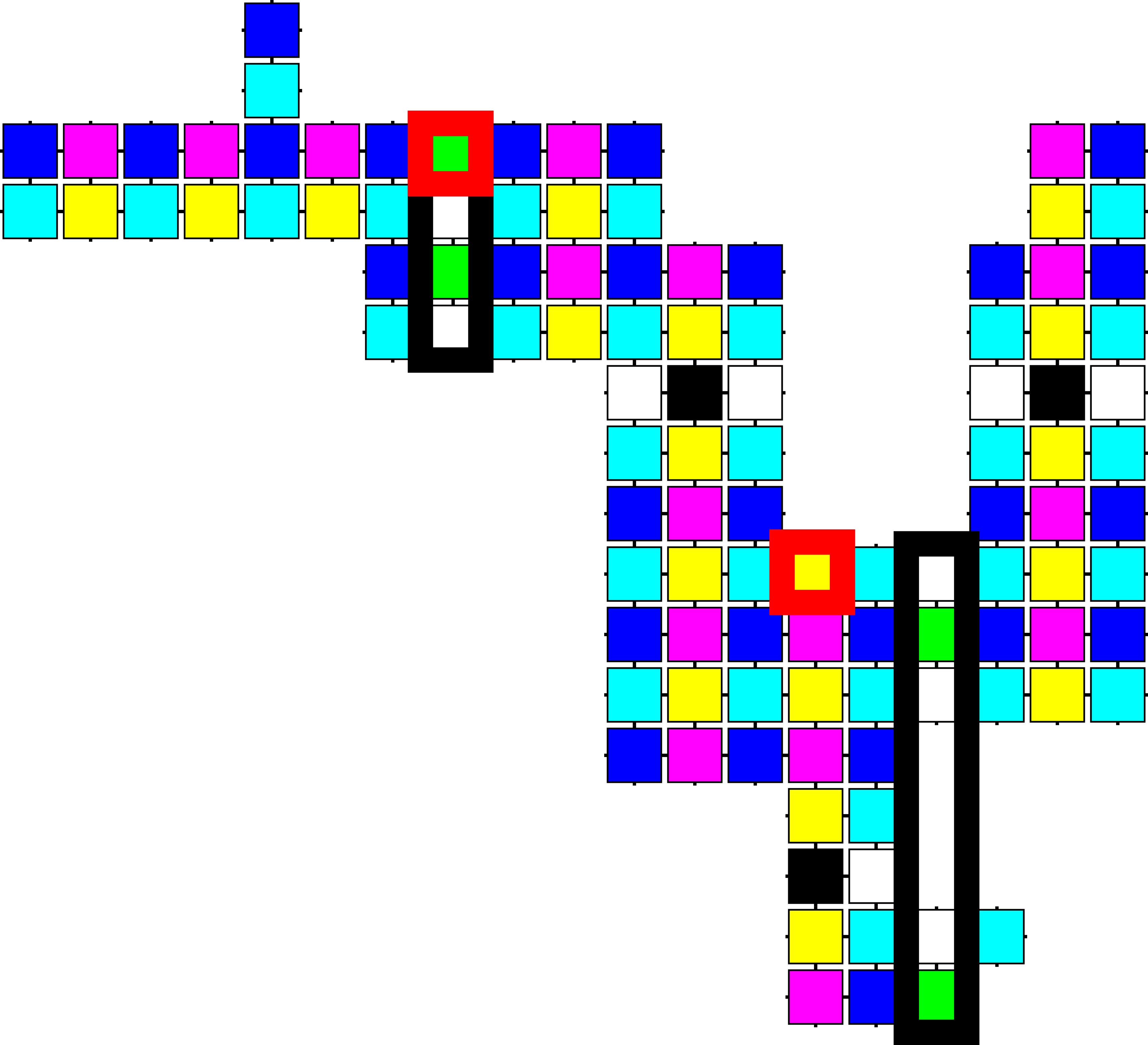}
    \caption{An example assembly $\alpha_i$ possibly formed during the simulation of some $\calT_i$. The leftmost red square hightlights a leftmost tile with a boundary color, and the corresponding black rectangle highlights the rest of that column and thus the boundary of a (potential) cell. Assuming $\var{SF}(n) = 8$, i.e., grid cell sizes of 8, the rightmost black rectangle highlights the locations of tiles at the boundary of the next cell to the right. Assuming an index value of $i = 6$, the rightmost red square highlights a tile at that index, with respect to the boundary to the left. The color of the highlighted cell is \texttt{Yellow}, which is one of the two colors reserved for columns representing the bit value 0. Therefore, the bit value 1 is saved for index $i$ to ensure that the pattern $p_n$ will never place a \texttt{Yellow} tile $i$ locations to the east of a tile with a boundary color, guaranteeing that $p_n$ differs from the pattern produced by $\calT_i$.}
    \label{fig:assembly-indexed}
\end{figure}

The simulation of $M$ proceeds through the simulation of each of the $\var{SF}(n)$ possible aTAM systems with $\le n$ tile types, $\le 8$ colors, and single-tile seeds (once again, just saving $0$s for strength-free systems that do not have corresponding aTAM systems). As the rows representing each simulation grow, they pass the currently computed sequence of bits, $b_n$, upward through the tiles performing the simulations. Once $M$ completes, $b_n$ is encoded in the north glues of the leftmost $\var{SF}(n)$ tiles of the top row. This is represented in Figure \ref{fig:wedge} as the \var{Black} and \var{White} sequence on the top of the aqua-colored portion of the wedge. Note that the tile types that simulate $M$ are a constant-sized tile set, regardless of the value of $n$.

At that point, another constant-sized set of tiles grow in a zig-zag manner to copy $b_n$ over and over, to the right, until the entire top row consists of copies of $b_n$ (with the last copy of the sequence potentially truncated). This is depicted as the yellow and (dark) grey portions of Figure \ref{fig:wedge}. Note that during all of the diagonal upward growth, a single ``filler'' tile type attaches to the right of the diagonal so that once the northward growth completes the full assembly will be a square. (This is shown as light grey in Figure \ref{fig:wedge}.)

\begin{figure}
    \centering
    \includegraphics[width=3.5in]{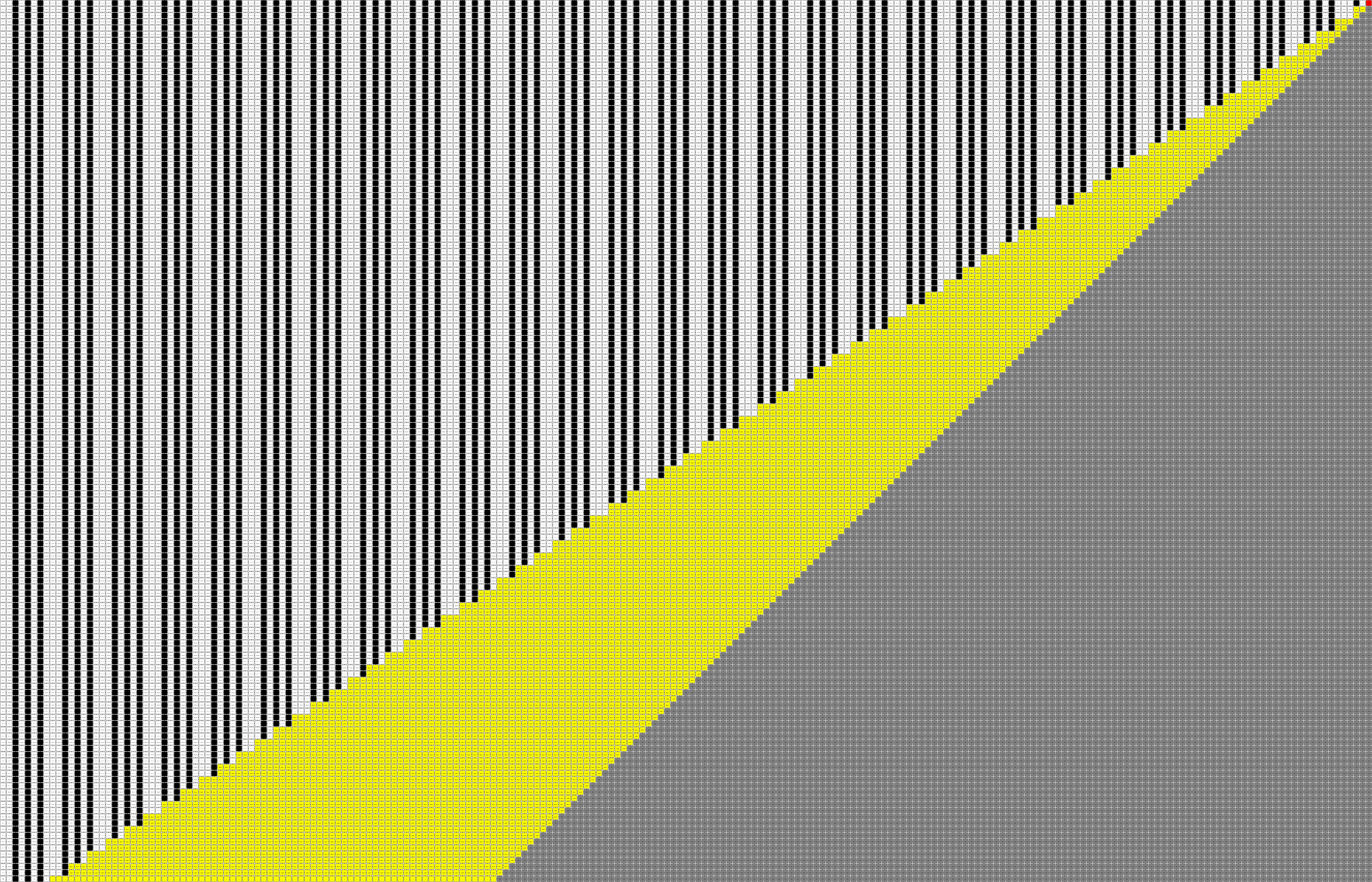}
    \caption{An example portion of the assembly from the proof of Theorem \ref{thm:2layers-square} that shows the copying of the bit sequence $b_n = 11010101$ to the right until it occupies the entire top row. This corresponds to the portion of the construction shown in Figure \ref{fig:wedge} above the aqua portion. At the completion of the bit sequence copying, the rightmost tile of the topmost row (shown here in red), initiates the growth into the second plane in which the pattern $p_n$ will assemble based on the bits of $b_n$. The grey portion below the yellow is formed by a single ``filler'' tile that causes the final assembly to be a square.}
    \label{fig:pattern-ext}
\end{figure}

Once the bit sequence $b_n$ has been copied across the entire northern row, the final phase of the construction begins. The easternmost tile to attach to the top row has a strength-2 glue in the $+z$ direction, initiating growth of the second plane, onto which the pattern $p_n$ self-assembles. The first row to grow in $z=1$ is immediately above the northernmost row of the assembly at $z=0$ and cooperates with the tiles of that row to read the repeated copies of $b_n$. This row forms the northern boundary row of all of the cells of $p_n$, and thus the tiles have colors from the set of boundary colors. The northernmost row in $z=0$ grows from the left to the right and includes information about the first (i.e., leftmost) bit of $b_n$. If that first bit is 0, the first row in $z=1$, as is a boundary row, contains the two boundary colors for a row of value 0, which are $\{\var{Red},\var{Green}\}$. Otherwise, it contains the two boundary colors for a row of value 1, which are $\{\var{White},\var{Black}\}$. If it was 0 (resp. 1), then as that first row in $z=1$ grows from right to left, when it cooperates with a tile representing 0 in $z=0$ it will be colored \var{Red} (resp. \var{White}), When cooperating with a tile representing a 1, it will be colored \var{Green} (resp. \var{Black}). (This is because each tile's color represents the combination of the row's bit value and the column's bit value.) The tile set that accomplishes the copying of $b_n$ across the entire top row in $z=0$ (shown as yellow and grey in Figure \ref{fig:wedge} and also in Figure \ref{fig:pattern-ext}) and grows the first row in $z=1$ consists of 1474 tile types. (Example tile assembly systems exhibiting some of the behavior described, as well as software that can generate them and simulate them, can be found online \cite{PatternAssemblySoftware}.)

\begin{figure}
    \centering
    \includegraphics[width=2.0in]{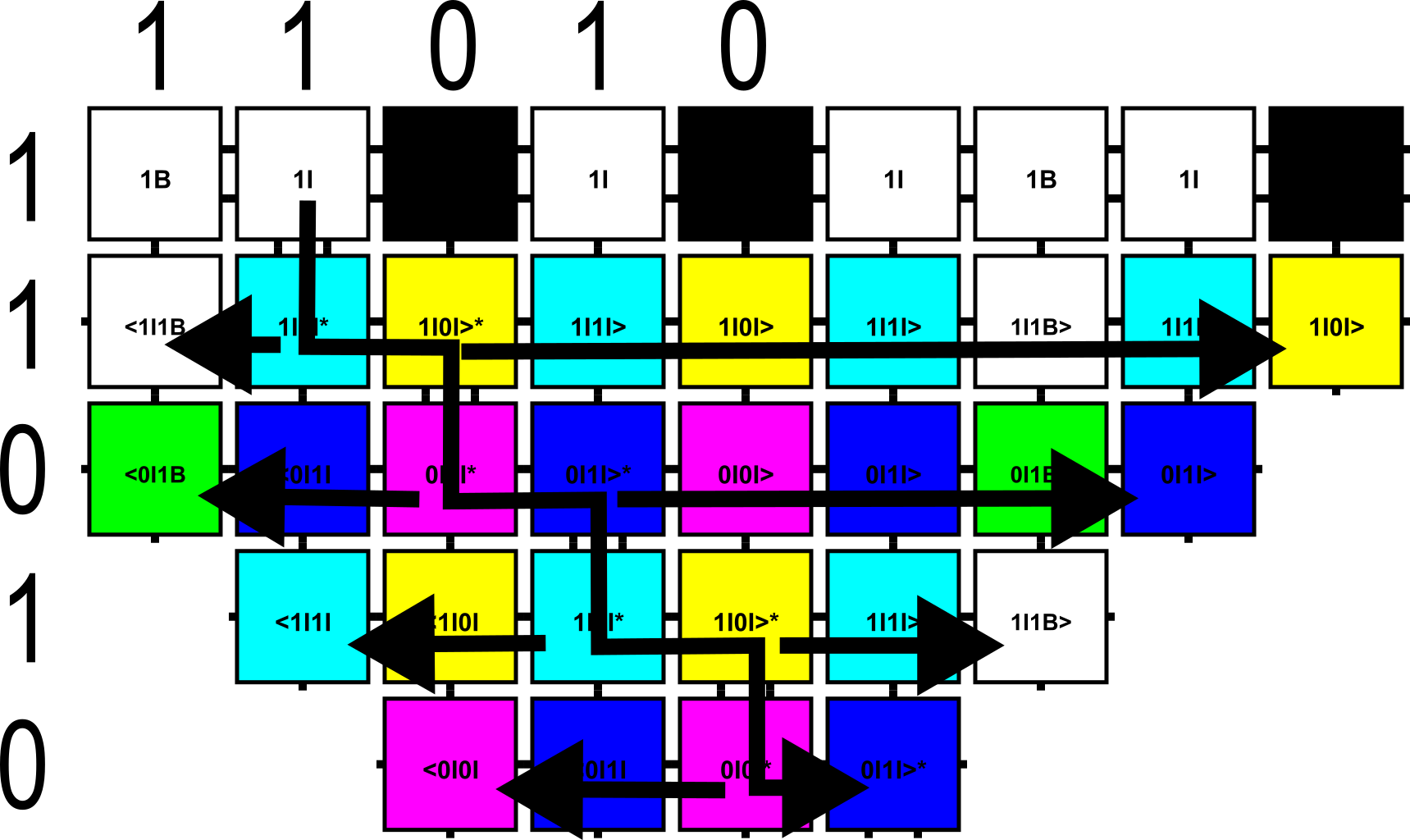}
    \caption{The growth of the grid, forming pattern $p_n$, from the topmost row, which is the initial row in plane $z=1$. The first tile placed in each row follows a diagonal path starting from the second tile from the left. Growth of each row expands left and right from the tile along the diagonal. The bit value for each column is propagated south from the top column, and the bit value for each row is propagated from the diagonal tile that initiates that row. In this way, the $i$th row propagates the same bit value as the $i$th column, and each tile's color represents the values of the pair of bits.}
    \label{fig:grid-growth}
\end{figure}

Upon completion of the first row of $z=1$, the final module of the construction begins growth. This module consists of 52 tile types that grow south from that first row to make the full square in $z=1$ while copying the pattern downward and to both sides to form the grid pattern $p_n$. (See Figure \ref{fig:grid-growth} for an example.)

\subsection{Correctness of proof}

Throughout the definition of the construction, we have explained the correctness of each component, so in this section we summarize those arguments to complete the proof of Theorem \ref{thm:2layers-square}. 
We first note that, by definition, our construction creates a pattern in the plane $z=1$ using 7 colors: 3 for the boundary rows and columns of each cell, and 4 for the interior regions of each cell. Which 3 boundary colors are used for any particular $p_n$ depends upon the result of the first system simulated, and are chosen from a set of 4 possible colors based on that bit value.

Next, we argue that for every $n$ the barely-3DaTAM system $\calT_{p_n}$ simulates every possible aTAM system with $\le n$ tile types, $\le 8$ colors, and a single-tile seed, or an equivalent aTAM system.
The details of how the enumeration of all such systems is guaranteed are given in Section \ref{sec:simulation}.
In summary, using the results of \cite{chen2015program} we know that all strength-free systems with $\le n$ tile types, $\le 8$ colors, and single-tile seeds are enumerated and then that all valid aTAM systems with $\le n$ tile types, $\le 8$ colors, and single-tile seeds that are equivalent to those are generated and simulated. By the definition of strength-free systems, this set will include an aTAM system from every set of equivalent aTAM systems fitting those criteria. Therefore, for any aTAM system $\calT_i$ with $\le n$ tile types, $\le 8$ colors, and a single-tile seed, a bit will be added to $b_n$ that will be mapped to a color differing from a tile placed by $\calT_i$ in a location specific to that value of $i$.
This results in $p_n$ differing from each of those systems in at least one location of each grid cell.
The fact that a bit is gathered for each simulation to guarantee that $p_n$ differs from it in at least one location is shown in Section \ref{sec:differ}.
Thus, it is shown that $p_n$ must differ from the pattern made by every aTAM system with $\le n$ tile types.

The value of $m$, which becomes the dimensions of the $m \times m$ $2$-layered square formed by $\calT_{p_n}$ is dependent upon the combined heights of the components that grow in $z=0$, which can be seen in Figure \ref{fig:wedge}. These heights are dominated by the runtime of the Turing machine $M$ and result in $m = O(n^{21n})$.

Finally, we just need to show that the tile complexity of the barely-3DaTAM system $\calT_{p_n}$ is $O(\log{n}/\log{\log{n}})$.
For this, we note that the tiles of all components are constant with respect to $n$, with the exception of the initial component that unpacks the value of $n$ from its optimal encoding using $O(\log{n}/\log{\log{n}})$ tile types, for an overall tile complexity of $O(\log{n}/\log{\log{n}})$. Thus, Theorem \ref{thm:2layers-square} is proved. (Nonetheless, for the interested reader we provide more details of the pseudocode of $M$ in Section \ref{sec:pseudocode} and a complexity analysis in Section \ref{sec:2layer-complexity}.)

\subsection{Pseudocode for $M$}\label{sec:pseudocode}

During the growth of the layer at $z=0$, the majority of the construction's complexity lies in the simulation of Turing machine $M$ that simulates every aTAM system with $\le n$ tile types, $\le 8$ colors, and single-tile seeds. Here we provide details of the how $M$ accomplishes that.

We break the functionality of $M$ into pieces for which we define the pseudocode. The main function executed by $M$ is $\var{SimulateAllTileAssemblySystems}$ that takes the maximum number of tile types, $n$, in the systems to be simulated and can be seen in Algorithm \ref{alg:sim-all}. This function computes the number of systems to simulate, initializes the data structure used to contain the tile set definitions, then loops to simulate each system and retrieve the relevant bit value needed from each to construct bit sequence $b_n$ for pattern $p_n$.

\begin{algorithm}
  \caption{An algorithm for generating all strength-free systems with $\le \var{numTileTypes}$ tile types of up to 8 colors with single-tile seeds, then converting each to an equivalent aTAM system (when possible) and simulating each aTAM system for a bounded amount of time. It returns a list of bits, one bit for each simulation.}\label{alg:sim-all}

  \begin{algorithmic}[1]
    \Procedure{SimulateAllTileAssemblySystems}{\var{numTileTypes}}
        \State $\var{B} = []$ \algorithmiccomment Initialize the list of bits

        \State $\var{numCoopSets} = 168$ \algorithmiccomment See Section \ref{sec:simulation} for details
        \State $\var{numColors} = 8$

        \State $\var{numSystems} = $ \Call{CountNumSFSystems}{$\var{numTileTypes}$}
        \State $\var{index} = 0$

        \ForEach{$\var{currNumTileTypes} \in [1,2,\ldots,\var{numTileTypes}]$}

            \State $\var{numGlues} = \var{currNumTileTypes} + 1$ \algorithmiccomment Allow \emph{null} glue
            
            \State $\var{numPossibleTileTypes} = \var{numCoopSets}*\var{numColors}*\var{numGlues}^4$
        
            \State $\var{numPossibleTileSets} = \var{numPossibleTileTypes}^\var{currNumTileTypes}$

            \State $\var{currTileSet}$ = \Call{InitializeSFTileSet}{$\var{currNumTileTypes}$}

            \ForEach{$\var{i} \in [1,2,\ldots,\var{numPossibleTileSets}]$}

                \ForEach{$\var{j} \in [1,2,\ldots,\var{currNumTileTypes}]$}

                    \State $\var{seedTile} = \var{currTileSet}[j]$
                    
                    \State $\var{currSFsystem} = (\var{currTileSet},\var{seedTile})$

                    \State $\var{tas} =$ \Call{GetEquivalentATAMSystem}{$\var{currSFSystem}$} \algorithmiccomment Function 
                    \State \algorithmiccomment from \cite{chen2015program}

                    \If{$\var{tas} == \var{FALSE}$}
                        \State $\var{B} = \var{B} + [0]$
                    \Else
                        \State $\var{numSteps} = (2*\var{numSystems})^2$
                        
                        \State $\var{b}_i$ = \Call{SimulateATAMSystem}{$\var{tas}, \var{numSteps}, \var{numSystems}, \var{index}$}
        
                        \State $\var{B} = \var{B} + [\var{b}_i]$
                    \EndIf

                    \State $\var{index} = \var{index} + 1$

                \EndForEach

                \State \Call{IncrementSFTileSet}{$\var{currTileSet}, \var{numColors}, \var{numGlues}$}
        
            \EndForEach
        
        \EndForEach

        \State \Return $\var{B}$
    \EndProcedure

  \end{algorithmic}
\end{algorithm}

$\var{SimulateAllTileAssemblySystems}$ utilizes a number of helper functions. The first is $\var{CountNumSFSystems}$, which can be seen in Algorithm \ref{alg:count-sfsystems}. This function loops over each tile set size from $1$ to $n$ (which is supplied as the argument named $\var{numTileTypes}$) and sums all possible strength-free systems with those numbers of tile types. Note that the counting of strength-free systems in this function and also in $\var{SimulateAllTileAssemblySystems}$ leads to duplicate copies of tile sets being created. Since it naively iterates through all possible combinations of tile types, there will be tile sets that (1) have multiple copies of the same tile type, and (2) have the same tile types as each other but simply in different orderings. All such duplicate systems will be equivalent to each other. Although they will each be simulated, this doesn't cause any problems with the construction. Since all equivalent systems will create the same assemblies as each other, the pattern $p_n$ will simply differ from the assemblies of each duplicated system in at least as many locations of every cell as there were equivalent versions of that system simulated. (Note that a more sophisticated version of the algorithm could remove such duplicates, but since it doesn't affect correctness and is much easier to understand, we have used this simple version.)

The second helper function is $\var{InitializeSFTileSet}$, which can be seen in Algorithm \ref{alg:initialize-tile-set}. It simply creates the list of 6-tuples, where each 6-tuple describes a tile type (its 4 glues, color, and cooperation set, as seen in Figure \ref{fig:strength-free-tile}).

\begin{algorithm}
  \caption{An algorithm to count all strength-free systems with $\le \var{numTileTypes}$ tile types of up to 8 colors with single-tile seeds.}\label{alg:count-sfsystems}

  \begin{algorithmic}[1]
    \Procedure{CountNumSFSystems}{\var{numTileTypes}}

        \State $\var{numCoopSets} = 168$ \algorithmiccomment See Section \ref{sec:simulation} for details
        \State $\var{numColors} = 8$

        \State $\var{count} = 0$

        \ForEach{$\var{currNumTileTypes} \in [1,2,\ldots,\var{numTileTypes}]$}

            \State $\var{numGlues} = \var{currNumTileTypes} + 1$ \algorithmiccomment Allow \emph{null} glue
            
            \State $\var{numPossibleTileTypes} = \var{numCoopSets}*\var{numColors}*\var{numGlues}^4$
        
            \State $\var{numPossibleTileSets} = \var{numPossibleTileTypes}^\var{currNumTileTypes}$

            \State $\var{numPossibleSFSystems} = \var{numPossibleTileSets} * \var{currNumTileTypes}$

            \State $\var{count} = \var{count} + \var{numPossibleSFSystems}$
        
        \EndForEach

        \State \Return $\var{count}$
    \EndProcedure

  \end{algorithmic}
\end{algorithm}

\begin{algorithm}
  \caption{A procedure to initialize the data structure representing a strength-free tile set.}\label{alg:initialize-tile-set}

  \begin{algorithmic}[1]
  \Procedure{InitializeSFTileSet}{\var{numTileTypes}}
    \State $\var{tileSet} = []$   \algorithmiccomment{Start with an empty list}
    \ForEach{$i \in [1,2,\ldots,\var{numTileTypes}]$}
        \State $\var{tile}_i = (0,0,0,0,0,0)$ \algorithmiccomment{Make the 6-tuple for a tile type}
        \State \algorithmiccomment{(N glue, E glue, S glue, W glue, color, coopSet)}

        \State $\var{tileSet} = \var{tileSet} + \var{tile}_i$
    \EndForEach
    \State \Return $\var{tileSet}$
  \EndProcedure

  \end{algorithmic}
\end{algorithm}

Next is the function used to increment to the next strength-free tile set to be tested. Called $\var{IncrementSFTileSet}$, this can be seen in Algorithm \ref{alg:incr-tile-set}.

\begin{algorithm}
  \caption{Algorithm to increment the current strength-free tile set to the next possible strength-free tile set.}\label{alg:incr-tile-set}

  \begin{algorithmic}[1]
  \Procedure {IncrementSFTileSet}{$\var{currTileSet}, \var{numColors}, \var{numGlues}$}
      \State $\var{overflow} = \var{TRUE}$
      \State $\var{tileNum} = 0$

      \While{$(\var{overflow} = \var{TRUE})$ and $(\var{tileNum} < len(\var{currTileSet}))$}

        \State $\var{tile} = \var{currTileSet}[\var{tileNum}]$

        \State $\var{currLoc} = 0$

        \While{$(\var{overflow} == \var{TRUE})$ and $(\var{currLoc} < 4)$}
            \If{$\var{tile}[\var{currLoc}] == (\var{numGlues} - 1)$}
                \State $\var{tile}[\var{currLoc}] = 0$
            \Else
                \State $\var{tile}[\var{currLoc}] += 1$
                \State $\var{overflow} = \var{FALSE}$
            \EndIf
            \State $\var{currLoc} += 1$
        \EndWhile

        \If{$\var{overflow} == \var{TRUE}$}
            \If{$\var{tile}[4] == (\var{numColors}-1)$}
                \State $\var{tile}[4] = 0$
            \Else
                \State $\var{tile}[4] += 1$
                \State $\var{overflow} == \var{FALSE}$
            \EndIf
        \EndIf

        \If{$\var{overflow} == \var{TRUE}$}
            \If{$\var{tile}[5] == 167$}
                \State $\var{tile}[5] = 0$
            \Else
                \State $\var{tile}[5] += 1$
                \State $\var{overflow} == \var{FALSE}$
            \EndIf
        \EndIf

        \State $\var{tileNum} += 1$
    \EndWhile
        
  \EndProcedure
  
  \end{algorithmic}
\end{algorithm}

The function $\var{SimulateATAMSystem}$ holds the logic for simulating each of the generated aTAM systems and also inspecting them to retrieve the necessary output bits. Its logic is shown in Algorithm \ref{alg:TAS-simulate}, and a high-level overview of the data structures used to store the current tile set, assembly and frontier can be seen in Figure ~\ref{fig:sim-encoding}. $\var{SimulateATAMSystem}$ also has a few helper functions to be discussed below.

\begin{figure}
    \centering
    \includegraphics[width=4.8in]{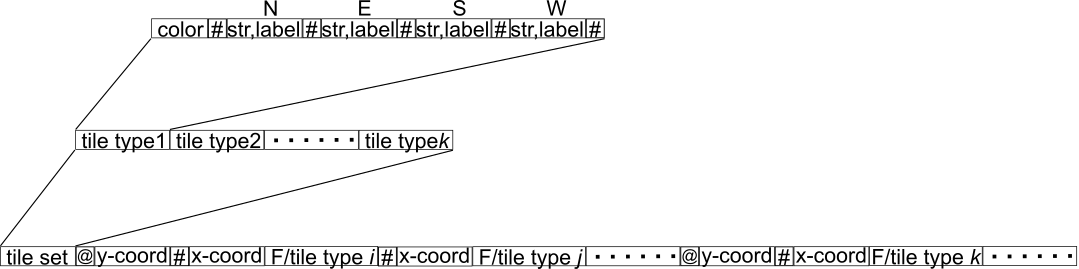}
    \caption{A high-level depiction of the encoding of the tile types, tile set, and assembly during a simulation of an aTAM system. (Top) The encoding of a tile type consists of its color, followed by the definition of each side's glue (i.e. its strength and label). (Middle) The encoding of a tile set consists of a list of each tile type's definition. (Bottom) The encoding of a system being simulated consists of the definition of the tile set followed by a sub-list for each $y$-coordinate containing a tile or frontier location, where each such sub-list consists of a list of entries of the $x$-coordinates at that $y$-coordinate containing tiles or frontier locations. Each such location contains the encoding of the $x$-coordinate and either the definition of the type of the tile located there or, for a frontier location, the definitions of the glues that are adjacent to it (along with a special character to denote the location as a frontier location).}
    \label{fig:sim-encoding}
\end{figure}

\begin{algorithm}
  \caption{Algorithm to simulate a given aTAM system for a bounded amount of steps and return a pattern value.}\label{alg:TAS-simulate}

  \begin{algorithmic}[1]
  
  \Procedure {SimulateATAMSystem}{$\var{tas}, \var{numSteps}, \var{pattSize}, \var{index}$}
    \State $\var{tileSet} = \var{tas}[0]$
    \State $\var{seed} = \var{tas}[1]$
    \State $F = \{(0,0)\}$ \algorithmiccomment Initialize the frontier
    \State $\alpha = \{\var{seed}\}$ \algorithmiccomment Initialize assembly as seed
    \State \Call{UpdateFrontier}{$F,(0,0),\alpha,\texttt{tileSet}$}
    \State $s = 0$
    \While{$s \le \var{numSteps}$}
        \State \Call{AddTile}{$F,\alpha,\texttt{tileSet}$}
        \If{$|F| = 0$}
            \State \Return $0$ \algorithmiccomment System doesn't make valid pattern
        \Else
            \State $s = s + 1$
        \EndIf
    \EndWhile
    \State $b = $\Call{GetPatternValue}{$\texttt{tileSet},\alpha, \var{pattSize}, \texttt{index}$} \algorithmiccomment Analyze pattern to
    \State \algorithmiccomment find return value $b$
    \State \Return $b$
  \EndProcedure

  \end{algorithmic}
\end{algorithm}

The first helper function for $\var{SimulateATAMSystem}$ is $\var{UpdateFrontier}$, which can be seen in Algorithm \ref{alg:update-frontier} and is used to update the current set of frontier locations after a tile is added to an assembly by first removing the location of the newly added tile from the frontier, then checking all locations neighboring that newly added tile to see if they need to be added to the frontier.

\begin{algorithm}
  \caption{A procedure that takes as arguments a set of frontier locations, one of the locations from that set, an assembly, and a tile set. It updates the frontier by adding any locations which neighbor the given location and have adjacent glues that would allow a tile of some type in the tile set to bind.}\label{alg:update-frontier}

  \begin{algorithmic}[1]
  \Procedure{UpdateFrontier}{$F,l,\alpha,T$}
        \State Remove location $l$ from $F$
        \For{$n \in \{(1,0),(-1,0),(0,1),(0,-1)\}$}
            \State $l_{nbr} = l + n$
            \State $\var{glues} = []$
            \If{$l_{nbr} \not \in \alpha$ and $l_{nbr} \not \in F$}
                \For{$n_2 \in \{(1,0),(-1,0),(0,1),(0,-1)\}$}
                    \State $l_{nbr2} = l_{nbr} + n_2$
                    \If{$l_{nbr2} \in \alpha$}
                        \State Let $t$ be the tile type at location $l_{nbr2}$ in $\alpha$
                        \State Let $g$ be the glue on the side of $t$ adjacent to $l_{nbr}$
                        \State $\var{glues} = \var{glues} + \var{g}$
                    \EndIf
                \EndFor
            \EndIf
            \If{Sum of strengths of glues in $\texttt{glues} \ge 2$}
                \For{$t \in T$}
                    \If{Sum of strengths of glues of $t$ matching glues in $\texttt{glues} \ge 2$}
                        \State $F = F \cup l_{nbr}$
                    \EndIf
                \EndFor
            \EndIf
        \EndFor
  \EndProcedure

  \end{algorithmic}
\end{algorithm}

The next helper function for $\var{SimulateATAMSystem}$ is $\var{AddTile}$, and is shown in Algorithm \ref{alg:add-tile}. It determines if and where a new tile can be added to an assembly, and adds one if possible.

\begin{algorithm}
  \caption{A procedure that takes as arguments a set of frontier locations, an assembly, and a tile set, then places a fitting tile type into the first listed frontier location (if the frontier is not empty). It ensures the frontier and assembly are correctly updated to account for the added tile.}\label{alg:add-tile}

  \begin{algorithmic}[1]
  
  \Procedure{AddTile}{$F,\alpha,T$}
    \If{$|F| = 0$}
        \State \Return \algorithmiccomment Empty frontier, can't add a tile
    \Else
        \State $f = F[0]$ \algorithmiccomment Get the first location in the frontier

        \ForEach{$t \in T$}
            \If{$t$ can bind in $f$}
                \State $\alpha = \alpha + (t,f)$ \algorithmiccomment Add a tile of type $t$ in location $f$
                \State $F = F - f$ \algorithmiccomment Remove $f$ from the frontier
                \State \Call{UpdateFrontier}{$F,f,\alpha,T$}
                \State \Return
            \EndIf
        \EndForEach
    \EndIf
  \EndProcedure

  \end{algorithmic}
\end{algorithm}

The function $\var{GetPatternValue}$ inspects a given assembly to find the color of a tile at a location matching the current system's index and, if found, returns a bit that will cause pattern $p_n$ to always have different colors at that index in the cells of the pattern. It simply returns 0 if the assembly fails to make a valid cell of a pattern.

\begin{algorithm}
  \caption{Procedure for analyzing an assembly and finding the color of a tile at a given index in the pattern.}\label{alg:assembly-analyzer}

  \begin{algorithmic}[1]
  
  \Procedure {GetPatternValue}{$\var{tileSet}, \alpha, \var{pattSize}, \var{index}$}
    \State Let $\var{BoundaryColors} = \{\texttt{Red},\texttt{Green},\texttt{Black},\texttt{White}\}$
    \State Let $\var{InteriorColors} = \{\texttt{Pink},\texttt{Blue}\,\texttt{Yellow},\texttt{Aqua}\}$
    \State $(\var{firstTileType}, (\var{firstX}, \var{firstY})) = \alpha[0]$ \algorithmiccomment Get first tile in assembly list
    \State $\var{minX} = \var{firstX}$ \algorithmiccomment Initialize min/max variables
    \State $\var{maxX} = \var{firstX}$
    \State $\var{minY} = \var{firstY}$
    \State $\var{maxY} = \var{firstY}$
    
    \ForEach{$\var{tile} \in \alpha$} \algorithmiccomment Find min/max coordinates of $\alpha$
        \State $\var{currX} = \var{tile}[1][0]$
        \State $\var{currY} = \var{tile}[1][1]$
        \If{$\var{currX} < \var{minX}$}
            \State $\var{minX} = \var{currX}$
        \EndIf
        \If{$\var{currX} > \var{maxX}$}
            \State $\var{maxX} = \var{currX}$
        \EndIf
        \If{$\var{currY} < \var{minY}$}
            \State $\var{minY} = \var{currY}$
        \EndIf
        \If{$\var{currY} > \var{maxY}$}
            \State $\var{maxY} = \var{currY}$
        \EndIf
    \EndForEach
    
    \If{$\var{maxX} - \var{minX} \ge \var{pattSize}$} \algorithmiccomment $\alpha$ is wide enough to contain the pattern
        \State \Return \Call{InspectWidth}{$\alpha,\var{minX}, \var{maxX}, \var{index}$}
    \Else \algorithmiccomment $\alpha$ must be tall enough to contain the pattern
        \State \Return \Call{InspectHeight}{$\alpha,\var{minY}, \var{maxY}, \var{index}$}
    \EndIf
  \EndProcedure

  \end{algorithmic}
\end{algorithm}

The four helper functions for $\var{GetPatternValue}$ are $\var{InspectWidth}$, $\var{InspectHeight}$, $\var{GetTileAtX}$ and $\var{GetTileAtY}$, shown in Algorithms \ref{alg:inspect-width}, \ref{alg:inspect-height}, \ref{alg:x-coord}, and \ref{alg:y-coord}, respectively. The first two take as input an assembly and its greatest extent in the corresponding dimension, plus the $\var{index}$ of the current simulation, and look for a tile with a boundary color, then a tile further along the corresponding dimension by $\var{index}$ positions, and return a bit related to the color of the tile there. More specifically, based on the color of the tile found there, it will return a bit that will force $p_n$ to disagree on colors at all locations corresponding to that $\var{index}$ value. (If a tile with a boundary color is not found, 0 is returned since the given assembly clearly cannot make pattern $p_n$.) $\var{GetTileAtX}$ and $\var{GetTileAtY}$ are simply used to find a tile with the given coordinate in an assembly.

\begin{algorithm}
  \caption{A procedure that takes an assembly, its horizontal bounds, and the index of the system being simulated, and searches horizontally to return a bit that ensures $p_n$ will have a different color at the index location than the assembly (or 0 if the assembly does not contain a valid pattern).}\label{alg:inspect-width}
  \begin{algorithmic}[1]
  \Procedure{InspectWidth}{$\alpha,\var{minX}, \var{maxX}, \var{index}$}
    \State $\var{currX} = \var{minX}$
    \State $\var{xTile} =$ \Call{GetTileAtX}{$\alpha, \var{currX}$}
    \State $\var{currColor} = \var{xTile}.\var{color}$
    \While{($\var{currColor} \not \in \var{BoundaryColors}$) and ($\var{currX} < \var{maxX}$)}
        \State $\var{xTile} =$ \Call{GetTileAtX}{$\alpha, \var{currX}$}
        \State $\var{currColor} = \var{xTile}.\var{color}$
        \State $\var{currX} = \var{currX} + 1$
    \EndWhile
    \If{$\var{currX} == \var{maxX}$}
        \State \Return $0$ \algorithmiccomment System failed to make valid pattern
    \Else
        \State $\var{indexTile} = $ \Call{GetTileAtX}{$\alpha, \var{currX} + \var{index}$}
        \If{$\var{indexTile} == \var{FALSE}$}
            \State \Return $0$ \algorithmiccomment System failed to make valid pattern
        \Else   \algorithmiccomment $\var{currColor}$ is the color placed by this system at its unique index
            \If{$\texttt{index} == 0$} \algorithmiccomment Boundary column
                \If{$\var{currColor} \in \{\texttt{White},\texttt{Green}\}$} \algorithmiccomment Color of column value $1$
                    \State \Return $0$
                \Else \algorithmiccomment Color of column value $0$
                    \State \Return $1$
                \EndIf
            \Else
                \If{$\var{currColor} \in \{\texttt{Aqua},\texttt{Blue}\}$} \algorithmiccomment Color of column value $1$
                    \State \Return $0$
                \Else \algorithmiccomment Color of column value $0$
                    \State \Return $1$
                \EndIf
            \EndIf
        \EndIf
    \EndIf
  \EndProcedure
  \end{algorithmic}
\end{algorithm}

\begin{algorithm}
  \caption{A procedure that takes an assembly, its vertical bounds, and the index of the system being simulated, and searches vertically to return a bit that ensures $p_n$ will have a different color at the index location than the assembly (or 0 if the assembly does not contain a valid pattern).}\label{alg:inspect-height}
  \begin{algorithmic}[1]
  \Procedure{InspectHeight}{$\alpha,\var{minY}, \var{maxY}, \var{index}$}
  
        \State $\var{currY} = \var{minY}$
        \State $\var{yTile} =$ \Call{GetTileAtY}{$\alpha, \var{currY}$}
        \State $\var{currColor} = \var{yTile}.\var{color}$
        \While{($\var{currColor} \not \in \var{BoundaryColors}$) and ($\var{currY} < \var{maxY}$)}
            \State $\var{yTile} =$ \Call{GetTileAtY}{$\alpha, \var{currY}$}
            \State $\var{currColor} = \var{yTile}.\var{color}$
            \State $\var{currY} = \var{currY} + 1$
        \EndWhile
        \If{$\var{currY} == \var{maxY}$}
            \State \Return $0$ \algorithmiccomment System failed to make valid pattern
        \Else
            \State $\var{indexTile} = $ \Call{GetTileAtY}{$\alpha, \var{currY} + \var{index}$}
            \If{$\var{indexTile} == \var{FALSE}$}
                \State \Return $0$ \algorithmiccomment System failed to make valid pattern
            \Else   \algorithmiccomment $\var{currColor}$ is the color placed by this system at its unique index
                \If{$\texttt{index} == 0$} \algorithmiccomment Boundary column
                    \If{$\var{currColor} \in \{\texttt{White},\texttt{Black}\}$} \algorithmiccomment Color of column value $1$
                        \State \Return $0$
                    \Else \algorithmiccomment Color of column value $0$
                        \State \Return $1$
                    \EndIf
                \Else
                    \If{$\var{currColor} \in \{\texttt{Aqua},\texttt{Yellow}\}$} \algorithmiccomment Color of column value $1$
                        \State \Return $0$
                    \Else \algorithmiccomment Color of column value $0$
                        \State \Return $1$
                    \EndIf
                \EndIf
            \EndIf
        \EndIf

  \EndProcedure
  \end{algorithmic}
\end{algorithm}

\begin{algorithm}
  \caption{A procedure that takes as arguments an assembly and $x$-coordinate value and returns a tile with that coordinate.}\label{alg:x-coord}
  \begin{algorithmic}[1]
  \Procedure{GetTileAtX}{$\alpha,x$}
    \ForEach{$\var{tile} \in \alpha$}
        \State $\var{currX} = \var{tile}[1][0]$
        \If{$\var{currX} = x$}
            \State \Return $\var{tile}$
        \EndIf
    \EndForEach
    \State \Return $\var{FALSE}$
  \EndProcedure
  \end{algorithmic}
\end{algorithm}

\begin{algorithm}
  \caption{A procedure that takes as arguments an assembly and $y$-coordinate value and returns a tile with that coordinate.}\label{alg:y-coord}
  \begin{algorithmic}[1]
  \Procedure{GetTileAtY}{$\alpha,y$}
    \ForEach{$\var{tile} \in \alpha$}
        \State $\var{currY} = \var{tile}[1][1]$
        \If{$\var{currY} = y$}
            \State \Return $\var{tile}$
        \EndIf
    \EndForEach
    \State \Return $\var{FALSE}$
  \EndProcedure
  \end{algorithmic}
\end{algorithm}

\subsection{Complexity analysis}\label{sec:2layer-complexity}

Here we give a brief overview of the time complexity of Turing machine $M$ running on input $n$, which in turn determines the size of the $m \times m$ square formed for each pattern $p_n$.

As shown in Section \ref{sec:simulation}, given a bound of $n$ tile types, an upper bound on the number of possible strength-free systems is $(5376n^4)^{n+2} = O(n^{4n}n^8)$, which we refer to as $\var{SF}(n)$. The algorithm that generates an equivalent aTAM system from a given strength-free system of $n$ tile types (if one exists) requires time $O(n^5)$. Each aTAM system is simulated for $O(\var{SF}(n)^2)$ tile additions. This means that the total number of simulated steps is bounded by $O(n^{4n}n^8n^5(n^{4n}n^8)^2) = O(n^{4n}n^{13}n^{8n}n^{16}) = O(n^{12n}n^{29})$.

A tile set of $n$ tile types requires $O(n\log{n})$ bits to represent.
The assembly $\alpha_i$ of each simulated system $\calT_i$ is represented as a combined list of assembly and frontier locations. (See Figure \ref{fig:sim-encoding} for a high-level depiction.) This list will be separated into a sub-list for each $y$-coordinate that contains a tile and/or frontier location. The sub-list for each $y$-coordinate will consist of an entry for each $x$-coordinate such that the coordinate $(x,y)$ represents a location with a tile in $\alpha$ or is a frontier location. Each tile location contains a definition of the tile type located there, requiring $O(\log{n})$ bits, and each frontier location contains the definitions of any (up to a maximum of 4) glues that are adjacent to that location and their directions, also requiring $O(\log{n})$ bits. Without loss of generality, the seed tile is placed at $(0,0)$. Thus, the encoding of each $x$ or $y$ coordinate requires $O(\log{\var{SF}(n)^2})$ bits, since the largest magnitude of any coordinate values can be $\var{SF}(n)^2$ or $-\var{SF}(n)^2$ if the simulation proceeds for $\var{SF}(n)^2$ steps. This means that the encoding of each entry for an $x$-coordinate plus tile or frontier location requires $O(\log{\var{SF}(n)^2} + \log{n})$ bits. Since $\var{SF}(n) = O(n^{4n}n^8)$, $\var{SF}(n) >> n$, so $O(\log{\var{SF}(n)} + \log{n}) = O(\log{{SN}(n)}) = O(n\log{n})$.

There can be $O(\var{SF}(n)^2)$ tile and frontier entries representing an assembly, for a size of $O(n\log(n)\var{SF}(n)^2) = O(\log(n)n^{17}n^{8n})$. The addition of a tile and updating of the frontier requires $O(n\log(n))$ traversals of the assembly, yielding a time per simulation step of $O(\log(n)^2n^{18}n^{8n})$. With $O(n^{12n}n^{29})$ simulation steps, that yields a total run time of 
$O(\log(n)^2n^{47}n^{20n})$ for M.


Since the simulation of each step of $M$ requires 1 or 2 rows (depending on the direction that the head of $M$ must move and whether the next row of the simulation grows right-to-left or left-to-right), and each row increases in width by 1, this is the bound of both the height and width of the assembly once the module that simulates $M$ completes growth.

The final portion to grow in the $z=0$ plane is that which copies the pattern, which will be of length $\var{SF}(n)$, across the entire width of the top row by increasing the width of the copied pattern by 3 for every 2 rows that grow upward, which themselves increase the width by 2. That means that approximately as many rows as the width of the top row before copying begins minus the width of the pattern are required, i.e., $O(\log(n)^2n^{47}n^{20n}) - \var{SF}(n)) = O(\log(n)^2n^{47}n^{20n})$. Thus, this is the bound for the width and height of the assembly that grows in the plane $z=0$, making the $m \times m$ square for $m = O(\log(n)^2n^{47}n^{20n}) = O(n^{21n})$.

    \section{Technical Details of the Construction for Corollary \ref{cor:2layers-infinite}}\label{sec:multilayer-infinite-appendix}

Given that $p_n$ forms on the $z=1$ layer of an $m \times m$ square, once the initial square forms, this extension causes it to be branched off of in the four cardinal directions, starting from positions $(0,m-1)$ for southward and westward expansion, and $(m-1, 0)$ for northward and eastward expansion. These patterns assemble new copies of themselves upon two of three boundaries defined along their edges with glue signals, which are either propagated along from the previously formed square, or are created upon the diagonal grid-generating signal making contact with an associated boundary.

\begin{figure}
    \centering
    \begin{subfigure}{0.25\textwidth}
        \centering
        \includegraphics[width=1.0\textwidth]{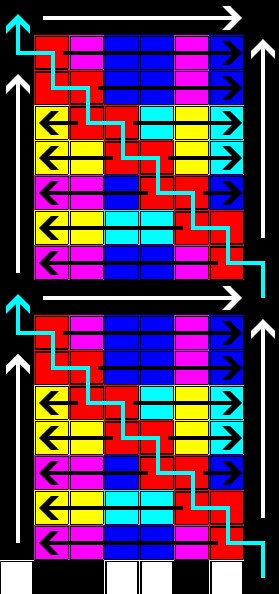}
        \caption{\label{fig:infinite-assembly-north}}
    \end{subfigure}
    \hspace{15pt}
    \begin{subfigure}{0.3\textwidth}
        \centering
        \includegraphics[width=1.0\textwidth]{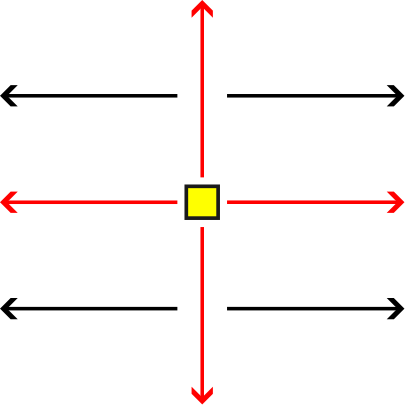}
        \caption{\label{fig:infinite-assembly-hgih-level}}
    \end{subfigure}
    \caption{(a) The growth of the grid, forming pattern $p_{north}$, which is identical to pattern $p_n$, with adjustments made for expansion to the north, from its bottom-most row, which is the initial row on the plane $z=1$. The first tile placed in each row follows a diagonal path, colored in red for clarity, starting from the first tile to the right. Growth of each row expands left and right from the tile along the diagonal. Upon the leftmost and rightmost edges are bounding tiles, colored in black for clarity, which denote the leftmost and rightmost bounds of $p_{north}$. Once the diagonal makes contact with the left bound, it generates a top-bound at height $h+1$, where $h$ is the height at which the diagonal makes contact with the left bound. Once the top-bound's signal is propagated to reach the right-bound, a new grid-forming diagonal is created, starting the process over again. Rotated copies of this process occur in each direction from the initial square of $p_n$. (b) High level illustration of the infinite growth of $m \times m$ squares, with the initial $m \times m$ square which generates $p_{n}$ colored in yellow, direction of growth colored in red, and direction of copy assembly, which utilizes the signals of those assembled above or below them colored in black.}
\end{figure}

This basic assembly is used in all four cardinal directions in order to form a cross of $m \times m$ squares, each of which extends infinitely. Once a column is formed in either the east or west direction, and a corresponding row is formed in the north or south direction, tiles containing the bit values of that column to their north or south, and rows to their east or west are able to propagate out along those patterns created in the cross assembly. This allows for the infinite tiling of the same $m \times m$ square infinitely across the $\mathbb{Z}^2$ plane to form $p_{n_\infty}$. 

\begin{figure}
    \centering
    \begin{subfigure}[T]{0.48\textwidth}
        \includegraphics[width=\textwidth]{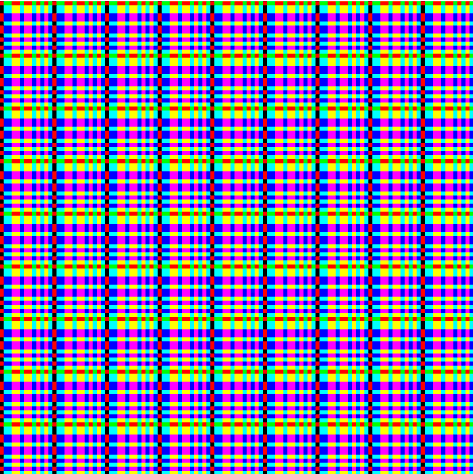}
        \caption{Portion of the infinite tiling of the bit sequence $0110010010101$.}
        \label{fig:infinite-tiling-unmarked}
    \end{subfigure}
    \hfill
    \begin{subfigure}[T]{0.48\textwidth}
        \includegraphics[width=\textwidth]{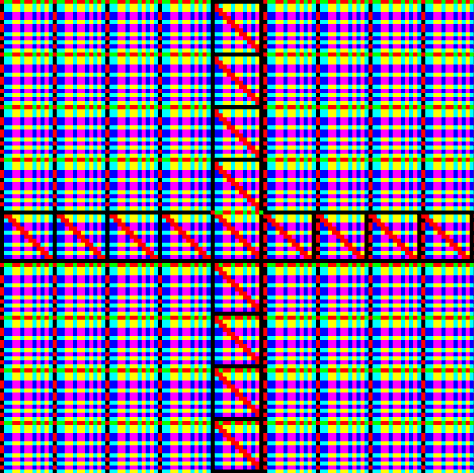}
        \caption{Portion of the infinite tiling of the bit sequence $0110010010101$ with the cross assembly marked with red tiles along grid-creating diagonals, and black tiles along bounds.}
        \label{fig:infinite-tiling-marked}
    \end{subfigure}
    \caption{Assemblies representing portions of the infinite tiling of $m \times m$ squares corresponding to bit sequence $0110010010101$ across $\mathbb{Z}^2$, with creation and bounds marked (right), and unmarked (left).}
    \label{fig:infinite-tiling-demonstration}
\end{figure}